\RequirePackage[hyphens]{url}
\documentclass[11pt]{article}
\usepackage[disable]{todonotes}
\usepackage[margin=1in]{geometry}
\usepackage[utf8]{inputenc}

\usepackage[hyphens]{url}
\usepackage{amsmath, amsthm, hyperref}
\usepackage{amsfonts}
\usepackage{amssymb, subfiles,nicefrac}

\newtheorem{theorem}{Theorem}[section]

\newtheorem{corr}[theorem]{Corollary}

\newtheorem{defn}[theorem]{Definition}
\newtheorem{claim}[theorem]{Claim}
\newtheorem{lemma}[theorem]{Lemma}
\newtheorem{fact}[theorem]{Fact}

\DeclareMathOperator{\samp}{sc}
\DeclareMathOperator{\scz}{\samp_{\rho}}
\DeclareMathOperator{\err}{err}

\DeclareMathOperator{\conv}{conv}
\DeclareMathOperator{\sign}{sign}

\newcommand{\cut}[1]{}

\DeclareMathOperator*{\argmin}{arg\,min}

\newcommand{\poly}{\mathrm{poly}}

\newcommand{\alg}{\mathcal{A}}

\newcommand{\uni}{\mathcal{X}}
\newcommand{\db}{X}%{\mathcal{D}}

\newcommand{\row}{x}
\newcommand{\quer}{\mathcal{Q}}

\newcommand{\Oh}{\mathcal{O}}
\newcommand{\eps}{\varepsilon}

\renewcommand{\Pr}{\mathbb{P}}
\newcommand{\E}{\mathbb{E}}
\newcommand{\rendiv}{\mathrm{D}}

\newcommand{\R}{\mathbb{R}}

\newcommand{\cov}{\mathcal{N}}
\newcommand{\sep}{\mathcal{P}}
\newcommand{\gmw}{\mathrm{g}}
\newcommand{\inprod}[1]{\langle #1 \rangle}

\renewcommand{\d}{\,\mathrm{d}}

\begin{document}
\title{\Large{Towards Instance-Optimal Private Query Release}}
\author{Jaros\l{}aw B\l{}asiok\thanks{ John A. Paulson School of Engineering and Applied Sciences, Harvard University, 33 Oxford Street,
Cambridge, MA 02138, USA. Email: {\tt jblasiok@g.harvard.edu.} } 
\and Mark Bun \footnote{Simons Institute for the Theory of Computing, \url{mbun@berkeley.edu}. Part of this work was done while the author was at Princeton University.} \and Aleksandar Nikolov
  \footnote{University of Toronto, \url{anikolov@cs.toronto.edu}} \and
  Thomas Steinke \footnote{IBM, Almaden Research Center, \url{opt-cdp@thomas-steinke.net}}}
\date{}
\pagenumbering{gobble}
\maketitle
\begin{abstract}
  We study efficient mechanisms for the query release problem in differential privacy: given a workload of $m$ statistical queries, output approximate answers to the queries while satisfying the constraints of differential privacy. In particular, we are interested in mechanisms that optimally adapt to the given workload. Building on the projection mechanism of Nikolov, Talwar, and Zhang, and using the ideas behind Dudley's chaining inequality, we propose new efficient algorithms for the query release problem, and prove that they achieve optimal sample complexity for the given workload (up to constant factors, in certain parameter regimes) with respect to the class of mechanisms that satisfy concentrated differential privacy. We also give variants of our algorithms that satisfy local differential privacy, and prove that they also achieve optimal sample complexity among all local sequentially interactive private mechanisms.
\end{abstract}

\pagebreak
\pagenumbering{arabic}
\setcounter{page}{1}

\section{Introduction}
\label{sect:intro}

There is an inherent tension in data analysis between privacy and
statistical utility. This tension is captured by the Fundamental Law
of Information Recovery: Revealing ``overly accurate answers to too
many questions will destroy privacy''.\footnote{This formulation is
  from~\cite{DR14-monograph}, while the quantitative version is
  from~\cite{DinurN03}.} This tension, however, is not equally
pronounced for every set of queries an analyst may wish to evaluate on a sensitive dataset. As a simple 
illustration, a single query
repeated $m$ times is much easier to answer while preserving privacy
than is a collection of $m$ random queries. For this reason, one of the basic goals of
algorithmic privacy research is to design efficient private algorithms
that optimally adapt to the structure of any given collection of queries. Phrased more specifically, this goal is to design 
algorithms achieving nearly optimal  \emph{sample complexity}---the minimum dataset size required to 
privately produce answers to within a prescribed accuracy---for any given workload of queries.
%and answer them within a
%prescribed accuracy on datasets that are as small as possible.

In this work, we address this problem in the context of answering
statistical queries in both the central and local models of
differential privacy.  Building on the projection mechanism
\cite{NTZ}, and using the ideas behind Dudley's chaining inequality,
we propose new algorithms for privately answering statistical
queries. Our algorithms are efficient and achieve instance optimal sample
complexity (up to constant factors, in certain parameter regimes) with
respect to a large class of differentially private
algorithms. Specifically, for every collection of statistical queries,
our algorithms provide answers with constant average squared error
with the minimum dataset size requirement amongst all algorithms
satisfying \emph{concentrated differential privacy (CDP)}
\cite{DBLP:journals/corr/DworkR16,BunS16}. We further show that our
algorithmic techniques can be adapted to work in the \emph{local
  model} of differential privacy, where they again achieve optimal
sample complexity amongst all algorithms with constant average squared
error.

\subsection{Background}

A dataset $\db$ is a multiset of $n$ elements from a data universe
$\uni$. A statistical (also referred to in the literature as ``linear'') query is specified by a function $q:\uni
\to [0,1]$. Overloading notation, the value of the
query $q$ on a dataset $\db$ is
\[
q(\db) = \frac1n \sum_{\row \in \db}{q(x)} \in [0,1],
\]
where dataset elements appear in the sum with their multiplicity. A
query workload $\quer$ is simply a set of $m$ statistical queries.  We
use the notation $\quer(\db)=(q(\db))_{q \in \quer} \in [0,1]^m$ for the vector
of answers to the queries in $\quer$ on dataset $\db$.

In the centralized setting in which the dataset $\db$ is held by a
single trusted curator, we model privacy by (zero-)concentrated
differential privacy. This definition was introduced by Bun and
Steinke~\cite{BunS16}, and is essentially equivalent to the original
definition of (mean-)concentrated differential privacy proposed by
Dwork and Rothblum~\cite{DBLP:journals/corr/DworkR16}, and closely
related to Mironov's R\'enyi differential privacy~\cite{Mironov17}. Before we state
the definition, we recall that two datasets $\db, \db'$ of size $n$
are \emph{neighboring} if we can obtain $\db'$ from $\db$ by replacing
only one of its elements with another element of the universe $\uni$.
\begin{defn} \label{def:cdp}
  A randomized algorithm $\alg$ satisfies $\rho$-zCDP if for any two neighboring
  datasets $\db$ and $\db'$, and all $\gamma \in (1,\infty)$, 
  \[
  \rendiv_\gamma(\alg(\db) \| \alg(\db')) \le \gamma\rho,
  \]
  where $\rendiv_\gamma$ denotes the R\'enyi divergence of order $\gamma$ measured in nats.
\end{defn}
For the definition of R\'enyi divergence and further discussion of
concentrated differential privacy, we refer the reader to
Section~\ref{sec:cdp}. For now, we remark that concentrated
differential privacy is intermediate in strength between ``pure''
($\varepsilon$-) and ``approximate'' ($(\varepsilon,
\delta)$-)differential privacy, in the sense that every mechanism
satisfying $\varepsilon$-differential privacy also satisfies
$\frac{\varepsilon^2}{2}$-zCDP, and every mechanism satisfying
$\rho$-zCDP also satisfies $(\rho + 2\sqrt{\rho \log(1/\delta)},
\delta)$-differential privacy for every $\delta > 0$
(see~\cite{BunS16}).  Our privacy-preserving techniques (i.e.,
Gaussian noise addition) give privacy guarantees which are most
precisely captured by concentrated differential privacy. In general,
concentrated differential privacy captures a rich subclass (arguably,
the vast majority) of the techniques in the differential privacy
literature, including Laplace and Gaussian noise addition, the
exponential mechanism \cite{McSherryT07}, sparse
vector~\cite{DR14-monograph}, and private multiplicative weights
\cite{HardtR10}. Crucially, concentrated differential privacy admits a
simple and tight optimal composition theorem which matches the
guarantees of the so-called ``advanced composition'' theorem for
approximate differential privacy~\cite{DworkRV10}. Because of these
properties, concentrated differential privacy and its variants has
been adopted in a number of recent works on private machine learning,
for example~\cite{PrivateDL,ParkFCW16,Lee17,PATE}.\todo{more?}

%\todo[inline]{Should we move local DP discussion here and treat CDP
%  and LDP on equal footing?}

%\todo[inline]{Yes, let's do it.}

We also study the local model of differential privacy, in which the
sensitive dataset $\db$ is no longer held by a single trusted curator,
but is instead distributed between $n$ parties where each party holds
a single element $\row_i$. The parties engage in an interactive
protocol with a potentially untrusted server, whose goal is to learn
approximate answers to the queries in $\quer$. Each party is
responsible for protecting her own privacy, in the sense that the
joint view of every party except $i$ should be $(\eps,
\delta)$-differentially private with respect to the input $x_i$. (See
Section~\ref{sect:LDP} for the precise details of the definition.)
Almost all industrial deployments of differential privacy, including
those at Google \cite{RAPPOR}, Apple \cite{AppleDP}, and Microsoft \cite{MicrosoftDP}, operate in the local model, and
it has been the subject of intense study over the past few
years~\cite{BassilySmith15,BassilyNST17,DJW-ASA}. 
While the local model of privacy offers stronger guarantees to
individuals, it is more restrictive in terms of the available privacy
preserving techniques. In particular, algorithms based on the
exponential mechanism in general cannot be implemented in the local
model \cite{KLNRS}. Nevertheless, we show that our algorithms can be relatively
easily adapted to the local model with guarantees analogous to the
ones we get in the centralized model. We believe this is evidence for
the flexibility of our approach.

% it comes
%at the price of increased sample complexity.

\medskip

%Our goal is to design mechanisms that adapt to the query workload they are
%answering, so that their error depends on some measure of the ``complexity'' of the
%queries. 
In order to discuss error guarantees for private algorithms, let us
first introduce some notation. We consider two natural measures of
error: average (or root-mean squared) error, and worst-case error. For
an algorithm $\alg$ we define its error on a query workload $\quer$
and databases of size $n$ as follows.
\begin{align*}
\err^2(\quer, \alg, n) &= \max_\db \left(\E \left[\sum_{q \in
    \quer}{\frac{(\alg(\db)_q - q(X))^2}{m} }\right]\right)^{1/2}
= \max_\db \left( \E\frac{1}{|\quer|}\|\alg(\db) -
  \quer(\db)\|_2^2\right)^{1/2},\\
\err^\infty(\quer, \alg, n) &= \max_\db ~\E\left[\max_{q \in \quer}{|\alg(\db)_q - q(X)|}\right]
= \max_\db \E\|\alg(\db) -
  \quer(\db)\|_\infty,
\end{align*}
where each maximum is over all datasets $\db$ of size $n$,
$\alg(\db)_q$ is the answer to query $q$ given by the algorithm $\alg$
on input $\db$, and expectations are taken with respect to the random
choices of $\alg$. The notation is used analogously in the local
model, with an interactive protocol $\Pi$ in the place of the
algorithm $\alg$.

\subsection{Main Results}

For the rest of this paper we will work with an equivalent formulation
of the query release problem which is more natural from the
perspective of geometric techniques, and will also ease our
notation. For a given workload of $m$ queries $\quer$, we can define
the set $S_\quer \subseteq [0,1]^m$ by $S_\quer = \{\quer(\row) : \row
\in \uni\}$. We can identify each data universe element $\row$ with
the element $\quer(\row)$ of $S_\quer$, so we can think of $\db$ as
just a multiset of elements of $S_\quer$; the true query answers
$\quer(\db)$ then just become the mean of the elements in $\db$. This
motivates us to introduce the \emph{mean point problem}: given a
dataset $\db$ of $n$ elements from a (finite) universe $\uni \subseteq
[0,1]^m$, approximate the mean $\bar{\db} = \frac1n \sum_{\row \in
  \db}{\row}$, where, as usual, the dataset elements are enumerated
with repetition. We assume that the algorithm is explicitly given the
set $\uni$. By analogy with the query release problem, we can
define measures of error for any given dataset $ \db$ by
\begin{align*}
\err^2(\db, \alg) &= \left(\E \frac1m \|\alg(\db) - \bar{\db}\|_2^2\right)^{1/2},\\
\err^\infty(\db, \alg) &= \E \|\alg(\db) - \bar{\db}\|_\infty.
\end{align*}
Similarly, we can define for any finite universe $\uni$ the error measures
\begin{align*}
    \err^2(\uni, \alg, n) & = \max_{\db \in \uni^n} ~ \err^2(\db, \alg),\\
    \err^\infty(\uni, \alg, n) & = \max_{\db\in\uni^n} ~ \err^\infty(\db, \alg).
\end{align*}
%The optimal
%sample complexity bounds $\samp_\rho^2(\uni, \alpha)$ and
%$\samp_\rho^\infty(\uni, \alpha)$ are then defined analogously to the
%definition for query release. From the discussion above it is clear
%that $\samp_\rho^2(\quer, \alpha) = \samp_\rho^2(S_\quer, \alpha)$ and
%$\samp_\rho^\infty(\quer, \alpha) = \samp_\rho^\infty(S_\quer,
%\alpha)$ for any workload $\quer$.
Algorithms for the query release problem can be used for the
corresponding mean point problem, and vice versa, with the same error
and privacy guarantees. Therefore, for the rest of the paper we will
focus on the mean point problem with the understanding that analogous
results for query release follow immediately.

In this work, we give query release algorithms whose error guarantees
naturally adapt to the properties of the queries, or, equivalently, we
give algorithms for the mean point problem that adapt to the geometric
properties of the set $\uni$.  Notice that for any datasets $\db$ and
$\db'$ of size $n$ that differ in a single element, $\bar{\db} -
\bar{\db'} \in \frac1n (\uni + (-\uni))$, where $\uni + (-\uni)$ is
the set of all pairwise sums of elements in $\uni$ and $-\uni$. Since
differential privacy should hide the difference between $\db$ and
$\db'$, a private algorithm should not reveal where in the set
$\bar{\db} + \frac1n (\uni + (-\uni))$ the true mean lies. This
suggests that that the size of $\uni + (-\uni)$, and, relatedly, the
size of $\uni$ itself, should control the sample complexity of the
mean point problem. However, it is non-trivial to identify the correct
notion of ``size'', and to design algorithms whose sample complexity
adapts to this notion. In this work we adopt separation numbers, which
quantify how many well-separated points can be packed in $\uni$, as a
measure of the size of $\uni$. More precisely, for any set $S
\subseteq \R^m$ and $t > 0$, we define the \emph{separation number}
(a.k.a.~\emph{packing number}) as
\[\sep(S, t) = \sup \{|T| : T \subseteq S, \quad \forall x, y \in T, x \ne y \implies \|x - y\|_2 > t\sqrt{m} \}.\]
That is, $\sep(S, t)$ is the size of the largest set of points in
$S$ whose normalized pairwise distances are all greater than
$t$. Analogously, we define the $\ell_\infty$ separation number
as 
\[\sep_\infty(S, t) = \sup \{|T| : T \subseteq S, \quad \forall
x, y \in T, x \ne y \implies \|x - y\|_\infty > t\}.\] 
Our bounds for
average error will be expressed in terms of $\sep(\uni, t)$, and the
bounds for worst case error will be expressed in terms of
$\sep_\infty(\uni, t)$. We will give algorithms whose sample
complexity is controlled by the separation numbers, and we will also
prove nearly matching lower bounds in the regime of constant error.

\subsubsection{Average Error}

We propose two new algorithms for private query release:
the Coarse Projection Mechanism and the Chaining Mechanism.  Both
algorithms refine the Projection Mechanism of Nikolov, Talwar, and
Zhang~\cite{NTZ}. Recall that the Projection Mechanism simply adds
sufficient Gaussian noise to $\bar{\db}$ in order to guarantee
differential privacy, and then projects the noisy answer vector back
onto the convex hull of $\uni$, i.e.~outputs the closest point to the
noisy answer vector in the convex hull of $\uni$ with respect to the
$\ell_2$ norm. Miraculously, this projection step can dramatically
reduce the error of the original noisy answer vector. The resulting
error can be bounded by the \emph{Gaussian mean width} of $\uni$,
which in turn is always at most polylogarithmic in the cardinality
of $\uni$.

Our first refined algorithm, which we call the Coarse Projection
Mechanism, instead projects onto a minimal $O(\alpha)$-cover of
$\uni$, i.e.~a set $T \subseteq \uni$ such that $T +
O(\alpha\sqrt{m})\cdot B_2^m$ contains $\uni$. (Here, ``+'' denotes
the Minkowski sum, i.e.~the set of all pairwise sums, and $B_2^m$ is
the unit Euclidean ball in $\R^m$.)  Since this cover is potentially a
much smaller set than $\uni$, the projection may incur less error this
way. The size of a minimal cover is closely related to the separation
number, and the separation numbers themselves are related to the
Gaussian mean width by Dudley's chaining inequality. We use these
connections in the analysis of our algorithm. The guarantees of the
mechanism are captured by the next theorem.
%\todo[inline]
%{Maybe it is worth stating explicitly bounds for projection mechanism in terms of gaussian mean-width, and 
%saying something about Dudley's inequality? Otherwise the statement of this theorem is quite surprising --- we 
%say that we project on a single $\alpha$-net, but the bound involves $\sup{t \sqrt{\log(\sep(\uni, t))}}$.}

\begin{theorem}\label{thm:ub-coarse}
  There exists a constant $C$ and a $\rho$-zCDP algorithm $\alg_{CPM}$
  (the Coarse Projection Mechanism) for the mean point problem, which
  for any finite $\uni \subseteq [0,1]^m$, achieves average error
  $\err^2(\uni, \alg_{CPM}, n) \le \alpha$ as long as
  \begin{equation}
    \label{eq:ub-coarse}
    n = \Omega \left(\frac{\log(1/\alpha)}{\alpha^2\sqrt{\rho}} \cdot \sup \left\{ t \cdot \sqrt{\log(\sep(\uni, t))} \ : \ t 
\ge \alpha / C\right\}\right).
  \end{equation}
  Moreover, $\alg_{CPM}$ runs in time $\poly(|\uni|, n)$.
\end{theorem}

We note that the sample complexity of the Coarse Projection Mechanism
can be much lower than that of the Projection Mechanism. For example,
consider a set $\uni$ defined as a circular cone with apex at the
origin, and a ball of radius $\alpha \sqrt{m}$ centered at
$(1-\alpha)\sqrt{m}e_1$ as its base. (Here $e_1$ is the first standard
basis vector.) Then, a direct calculation reveals that in order to
achieve average error $\alpha$, the Projection Mechanism  requires a
dataset of size at least $\Omega(\alpha^{-1} \sqrt{m})$. By contrast,
the Coarse Projection Mechanism would project onto the line
segment from the apex to the center of the base and achieve error
$\alpha$ with a dataset of size $O(\alpha^{-1})$. While in this
example $\uni$ is not finite, it can be discretized to a finite set
without significantly changing the sample complexity. 

%The proof of Theorem~\ref{thm:ub-coarse} relies on Dudley's chaining
%inequality, which we use to relate the Gaussian mean width of minimal
%cover $T$ to the family of separation numbers of $\uni$.
Inspired by the proof of Dudley's inequality, we give an alternative
Chaining Mechanism whose error guarantees are incomparable to the
Coarse Projection Mechanism. Instead of just taking a single cover of
$\uni$, we take a sequence of progressively finer covers $T_1, \ldots,
T_k$. This allows us to write $\uni$ as the Minkowski sum $\uni^1 +
\ldots + \uni^k + O(\alpha\sqrt{m})\cdot B_2^m$, where the diameter of
$\uni^i$ decreases with $i$, while its cardinality grows. We can then
decompose the mean point problem over $\uni$ into a sequence of $k$
mean point problems, which we solve individually with the projection
mechanism. The next theorem captures the guarantees of this mechanism.

\begin{theorem}\label{thm:ub-chain}
  There exists a constant $C$ and a $\rho$-zCDP algorithm $\alg_{CM}$
  (the Chaining Mechanism) for the mean point problem, which for any
  finite $\uni \subseteq [0,1]^m$ achieves average error $\err^2(\uni,
  \alg, n) \le \alpha$ as long as
  \begin{equation}
    \label{eq:ub-chain}
    n = 
    \Omega \left(\frac{\log(1/\alpha)^{5/2}}{\alpha^2\sqrt{\rho}} \cdot \sup \left\{ t^2 \cdot \sqrt{\log(\sep(\uni, t))} 
\ : \ t \ge \alpha / C\right\}\right).
  \end{equation}
  Moreover, $\alg$ runs in time $\poly(|\uni|, n)$.
\end{theorem}
%\todo[inline]{Jarek: I think in both upper bound claims, we want to look at $t \ge \alpha/4$ --- that is what 
%follows from the proofs below, if someone has a chance to double-check this calculation, that would be great.}
Notice that the sample complexity upper bound~\eqref{eq:ub-chain}
depends on a larger power of $\log(1/\alpha)$ than the
bound~\eqref{eq:ub-coarse}, but replaces $t$ in the supremum with
$t^2$. This change can only improve the latter term, as $\sep(\uni, t)
= 1$ for any $t \ge 1$.

\subsubsection{Instance-Optimality}
While our algorithms are generic, we show that for constant error, they achieve optimal sample complexity for \emph{any} given workload of queries. To be more precise about the instance-optimality of our results,
%The sample complexity of an algorithm $\alg$ with error $\alpha$ is defined as:
%\begin{align*}
%  \samp(\alg, \quer, \alpha) &= \min \{ n: \err(\quer, \alg, n) \leq
%  \alpha \},\\
 % \samp^\infty(\alg, \quer, \alpha) &= \min \{ n: \err^\infty(\quer, \alg, n)
%  \leq \alpha \}.
%\end{align*}
we define the sample complexity of answering a query workload $\quer$ with
error $\alpha$ under $\rho$-zCDP by
\begin{align*}
  \samp_\rho^2(\quer, \alpha) &= \min~ \{~n : \exists~ \rho\mbox{-zCDP } \alg ~\mbox{s.t.}~ \err^2(\quer, \alg, n) \leq \alpha ~ \},\\
  \samp_\rho^\infty(\quer, \alpha) &= \min~ \{~n :  \exists~ \rho\mbox{-zCDP } \alg ~\mbox{s.t.}~ \err^
\infty(\quer, \alg, n) \leq \alpha ~ \}.
\end{align*}
In the local model we analogously define $\samp_{\eps,
  \delta}^{2,loc}(\quer, \alpha)$ and $\samp_{\eps,
  \delta}^{\infty,loc}(\quer, \alpha)$ with the minimum taken over all
protocols satisfying $(\eps, \delta)$-local differential privacy. 

For context, let us recall the sample complexity in the centralized
model of some known algorithms, and how it compares to the best
possible sample complexity. For average error, the projection
mechanism~\cite{NTZ} can answer any workload $\quer$ with error at
most $\alpha$ under $\rho$-zCDP as long as
\[
n =\Omega\left(\frac{\sqrt{\log |\uni|}}{\alpha^2\sqrt{\rho}}\right).
\]
It is known that there exist workloads $\quer$ for which this bound on
$n$ matches $\samp_\rho(\quer, \alpha)$ up to constant factors. One
particularly natural example here is the workload of $2$-way marginals
on the universe $\uni=\{0, 1\}^d$, which consists of $m={d \choose 2}$
queries~\cite{BunUV14}. %(a more direct proof is available for zCDP).
Thus, the sample complexity of private query release with respect to worst-case
workloads of any given size is well-understood. However, we know much
less about optimal mechanisms and the behavior of $\scz(\quer, \alpha)$ for specific workloads $\quer$. This behavior can depend
strongly on the workload. For example, for the workload of threshold
queries $\quer = \{q_t\}_{t = 1}^m$ defined on a totally ordered universe $\uni
= \{1, \ldots, m\}$ by $q_t(\row) = \mathbf{1}\{\row < t\}$, we have sample
complexity only $\scz(\quer, \alpha) = \tilde{O}(\frac1\alpha)$. This
motivates the following problems:
\begin{enumerate}
\item Characterize $\scz(\quer, \alpha)$ in terms of natural
  quantities associated with $\quer$.
\item Identify efficient algorithms whose sample complexity on any
  workload $\quer$ nearly matches the optimal sample complexity
  $\scz(\quer, \alpha)$. 
\end{enumerate}
We call algorithms with the property in item 2.\ above
\emph{approximately instance-optimal}. Note that it is a priori not
clear that there should exist any \emph{efficient} instance optimal
algorithms. Here our notion of efficiency is polynomial time in $n$,
the number of queries, and the size of the universe $\uni$. This is
natural, as this is the size of the input to the algorithm, which
needs to take a description of the queries in addition to the
database. One could wish for a more efficient algorithm when the
queries are specified implicitly, for example by a circuit, but this
has been shown to be impossible in general under cryptographic
assumptions~\cite{DworkNRRV09}.

We prove lower bounds showing instance optimality for our algorithms
when the error parameter $\alpha$ is constant. Once again, we state
the lower bounds for the mean point problem, rather than the query
release problem. The equivalence of the two problems implies that we
get the same optimality results for query release as for the mean
point problem. To state the results we extend our notation above to
the mean point problem, and define 
\begin{align*}
  \samp_\rho^2(\uni, \alpha) &= \min~ \{~n : \exists~ \rho\mbox{-zCDP } \alg ~\mbox{s.t.}~ \err^2(\uni, \alg, n) \leq \alpha ~ \},\\
  \samp_\rho^\infty(\uni, \alpha) &= \min~ \{~n :  \exists~ \rho\mbox{-zCDP } \alg ~\mbox{s.t.}~ \err^\infty(\uni, \alg, n) \leq \alpha ~ \},
\end{align*}
and, analogously for $\samp_{\eps, \delta}^{2,loc}(\uni, \alpha)$ and
$\samp_{\eps, \delta}^{\infty,loc}(\uni, \alpha)$. Building on the
packing lower bounds of Bun and Steinke~\cite{BunS16}, we show that
separation numbers also provide lower bounds on $\samp^2_\rho(\uni,
\alpha)$. The following theorem is proved in the appendix.
\todo{Move it from the appendix?}

\begin{theorem}\label{thm:lb}
For any finite $\uni \subseteq [0,1]^m$ and every $\alpha, \rho > 0$, we have
\begin{align}
  \label{eq:lb}
  \samp^2_\rho(\uni, \alpha) &\ge \Omega \left(\frac{1}{\alpha \sqrt{\rho}} \cdot \sup\left\{ t \cdot 
\sqrt{\log(\sep(\uni, t))} : t \ge 4\alpha \right\}\right).%;\\
%  \label{eq:lb-infy}
 %   \samp_\rho^\infty(\quer, \alpha) &\ge \Omega \left(\frac{1}{\alpha \sqrt{\rho}} \cdot \sup\left\{ t \cdot 
%\sqrt{\log(\sep_\infty(\uni), t)} : t \ge 4\alpha \right\}\right).
\end{align}
\end{theorem}

Comparing the lower bound \eqref{eq:lb} with our algorithmic
guarantees \eqref{eq:ub-coarse} and \eqref{eq:ub-chain}, we see that
the algorithms in Theorems~\ref{thm:ub-coarse}~and~\ref{thm:ub-chain}
can achieve error $\alpha$ on databases of size at most
$\tilde{O}(\frac{1}{\alpha}) \cdot \samp^2_\rho(\uni, \alpha/C')$,
where $\tilde{O}$ hides factors polynomial in $\log(1/\alpha)$ and
$C'$ is a constant. In other words, when the error $\alpha$ is
constant, our mechanisms have sample complexity which is
instance-optimal up to constant factors. The constant error regime is
practically the most interesting one and is widely studied in the
differential privacy literature. It captures the natural problem of
identifying the smallest database size on which the mean point problem
(resp.\ the query release problem) can be solved with non-trivial
error. In his survey~\cite{Vadhan17} Vadhan asked explicitly for a
characterization of the sample complexity of counting queries in
the constant error regime under approximate differential privacy (Open
Problem 5.25). Our results make a step towards resolving this
question by giving a characterization for the rich subclass of
algorithms satisfying concentrated differential privacy.

Beyond the constant error regime, proving instance optimality results
with tight dependence on the error parameter $\alpha$ remains a
tantalizing open problem.  We note that we are not aware of any $\uni$
for which the sample complexity of the Chaining Mechanism is
suboptimal by more than a $(\log(1/\alpha))^{O(1)}$ factor.

\subsubsection{Worst-Case Error}

Using a variant of the chaining mechanism from Theorem~\ref{thm:ub-chain},
we get a guarantee for worst-case error as well.

\begin{theorem}\label{thm:ub-infty}
  There exists constant $C$, and a $\rho$-zCDP algorithm $\alg$ that for any finite $\uni
  \subseteq [0,1]^m$ achieves $\err^\infty(\uni, \alg, n) \le \alpha$ as long as 
  \begin{equation}
  \label{eq:ub-infty}
  n = 
  \Omega \left(\frac{\log(m)\log(1/\alpha)^{5/2}}{\alpha^2\sqrt{\rho}} \cdot \sup \left\{ t^2 \cdot \sqrt{\log(\sep_
\infty(\uni, t))} \ : \ t \ge \alpha / C\right\}\right).
\end{equation}
Moreover, $\alg$ runs in time $\poly(|\uni|, n)$.
\end{theorem}
This result shows the flexibility of the chaining mechanism. The
analysis of the coarse projection mechanism relied crucially on
Dudley's inequality, which is tailored to Euclidean space and the
Gaussian mean width. There are, in general, no mechanisms with
worst-case error guarantees whose sample complexity depends on the
Gaussian mean width, so it is unclear how to adapt the coarse
projection mechanism to worst-case error. Nevertheless, by
incorporating the idea of chaining used in the proof of Dudley's
inequality inside the algorithm itself, we are able to derive an
analogous result.

A lower bound analogous to Theorem~\ref{thm:lb} for worst-case error
reveals that the sample complexity of the algorithm in
Theorem~\ref{thm:ub-infty} on workload $\quer$ with error $\alpha$ is
at most $\tilde{O}(\frac{\log(m)}{\alpha}) \cdot
\samp^\infty_\rho(\quer, \alpha/C)$. I.e., we get
instance-optimality up to a $O(\log m)$ factor for constant $\alpha$.

%The mechanism is analogous to
%the one in Theorem~\ref{thm:ub-recur}, except we take $T_1, \ldots,
%T_k$ to be covers in the $\ell_\infty$ norm, and we answer the workloads
%$\quer_1, \ldots, \quer_k$ using the private multiplicative weights
%mechanism~\cite{HardtR10}.

\subsubsection{Local Differential Privacy}

Illustrating further the flexibility of our techniques, we show that
the Coarse Projection Mechanism and the Chaining Mechanism can be
adapted to the local model. The protocols we design are
non-interactive, with each party sending a single message to the
server, and satisfy pure $\eps$-local differential privacy. The
protocols are in fact very similar to our algorithms in the central model,
except that instead of Gaussian noise we use a variant of the local
mean estimation algorithm from \cite{DJW-ASA} to achieve
privacy. The other steps in the Coarse Projection and the Chaining
Mechanisms are either pre- or post-processing of the data and can be
adapted seamlessly to the local model.   

\begin{theorem}\label{thm:ldp}
  There exists a constant $C$ and a non-interactive $\eps$-LDP protocol $\Pi_{CPM}$ that for any finite $\uni
  \subseteq [0,1]^m$ achieves average error $\err^2(\uni, \Pi_{CPM}, n) \le \alpha$ as long as 
  \begin{equation}\label{eq:ldp-coarse-ub}
    n = \Omega \left(\frac{\log(1/\alpha)^2}{\alpha^4\eps^2} \cdot \sup \left\{ t^2 \cdot \log(\sep(\uni, t)) \ : \ t 
\ge \alpha / C\right\}\right).
  \end{equation}
Furthermore, there exists a non-interactive $\eps$-LDP protocol $\Pi_{CM}$ that achieves average error\\ $\err^2(\uni, \Pi_{CM}, n) \le \alpha$ as long as 
  \begin{equation}
    n = 
    \Omega \left(\frac{\log(1/\alpha)^{6}}{\alpha^4\eps^2} \cdot \sup \left\{ t^4 \cdot \log(\sep(\uni, t))
\ : \ t \ge \alpha / C\right\}\right).
  \end{equation}
Both protocols run in time $\poly(|\uni|, n)$.
\end{theorem}

Moreover, for constant average error $\alpha$, our algorithms achieve
instance-optimal sample complexity up to constant factors. This is
true even with respect to $(\eps, \delta)$-LDP algorithms permitting
``sequential'' interaction between parties (see Section~\ref{sect:LDP}
for details of the model). The theorem is proved in the appendix using
the framework of Bassily and Smith~\cite{BassilySmith15}

\begin{theorem}\label{thm:lb-ldp}
  For any finite $\uni \subseteq [0,1]^m$, every $\alpha > 0$, and any
  $\delta$ satisfying $0 <\delta < \frac{\alpha^2\eps^3}{C
    \log(|\uni|/\eps)}$ for a sufficiently large constant $C$, we
  have
  \begin{align}
    \label{eq:lb-ldp}
    \samp^{2,\text{loc}}_{\eps, \delta}(\uni, \alpha) &\ge 
    \Omega \left(\frac{1}{\alpha^2\eps^2} \cdot \sup\left\{ t^2 \cdot 
        \log(\sep(\uni, t)) : t \ge 6\alpha \right\}\right).
\end{align}
\end{theorem}

It is an interesting open problem to extend these instance optimality
results to worst-case error. While the lower bound extends in a
straightforward way, our mechanisms do not, as there is no analog of
the projection mechanism for worst-case error, and also no analog of
the multiplicative weights mechanism in the local model. Moreover, it
is known that in the local model packing lower bounds like these in
Theorem~\ref{thm:lb-ldp} can be exponentially far from the true sample
complexity with respect to worst-case error. For instance,
Kasiviswanathan et al.~\cite{KLNRS} showed that learning parities over
the universe $\uni = \{0,1\}^d$ has sample complexity exponential in
$d$ in the local model, and learning parities easily reduces to
answering parity queries with small constant worst-case error. At the
same time, packing lower bounds can only show a lower bound which is
polynomial in $d$. Thus, worst case error has substantially different
behavior from average error in the local model and requires different
techniques.

\subsection{Related Work} \label{sec:pure-char}

Instance-optimal query release was previously studied in a line of work that brought
the tools of asymptotic convex geometry to differential
privacy~\cite{HardtT10,BhaskaraDKT12,NTZ,Nikolov15,KattisN17}.
However, despite significant effort, completely resolving these
questions for approximate differential privacy appears to remain out of  reach
of current techniques. 

The papers~\cite{HardtT10,BhaskaraDKT12} focus
on pure differential privacy, and their results only apply for very
small values of $\alpha$, while here we focus on the regime of
constant $\alpha$.  A characterization for pure differential privacy with constant $\alpha$ is known~\cite{RothNotes, BunS16, Vadhan17} based on the same geometric quantities considered in this work. %We 
%discuss this chracterization in more detail in Section~\ref{sec:pure-char} after defining these quantities. 
When phrased in our language, these works show that for every constant
error parameter $\alpha$, the sample complexity of the mean point
problem with pure differential privacy is characterized up to constant
factors by the logarithm of an appropriate separation number of the
set $\uni$. The sample complexity lower bound follows from a packing
argument. Meanwhile, the upper bound is obtained by using the
\emph{exponential mechanism} of McSherry and Talwar~\cite{McSherryT07}
to identify a point in a minimal cover of $\uni$ which is as close as
possible to $\bar{X}$. Unlike our algorithms, this application of the
exponential mechanism runs in time super-polynomial in $\uni$. 

While we prove instance-optimality of our algorithms using similar
lower bound techniques (i.e., the generalization of packing arguments
to CDP from~\cite{BunS16}), our new algorithms appear to be completely
different. There is no known analogue of the exponential mechanism
that is tailored to achieve optimal sample complexity for CDP, and our
algorithms are instead based on the projection mechanism.

The papers~\cite{NTZ,Nikolov15} focus on approximate differential privacy,
and give results for the entire range of $\alpha$, but their bounds
are loose by factors polynomial in $\log |\uni|$. We avoid such gaps,
since for many natural workloads, such as marginal queries, $|\uni|$ is
exponential in the other natural parameters of $\quer$. The recent
paper~\cite{KattisN17} is also very closely related to our work, but
does not prove tight upper and lower bounds on $\scz(\quer, \alpha)$ for
arbitrary $\quer$.

%Separation numbers have previously been used to characterize the sample complexity of query release with 
%\emph{pure} differential privacy. (To the best of our knowledge, this characterization first appeared in Lectures 
%5 and 6 of Roth~\cite{RothNotes}, building on~\cite{HardtT10, BlumLR13, Roth10}. It was then
%developed more formally by Bun and Steinke~\cite{BunS16} and Vadhan~\cite{Vadhan17}.) Recall that CDP is 
%a relaxation of pure differential privacy, and hence CDP algorithms can typically achieve lower sample 
%complexity with qualitatively similar privacy guarantees.

\section{Preliminaries}
\label{sect:prelim}

In this section we define basic notation, state the definitions of
concentrated differential privacy and local differential privacy, and
state the known algorithms which will serve as building blocks for our
own algorithms. We also describe the geometric tools which will be
used throughout this paper.

\subsection{Notation}

We use the notation $A \lesssim B$ to
denote the existence of an absolute constant $C$ such that $A \le CB$,
where $A$ and $B$ themselves may depend on a number of
parameters. Similarly, $A \gtrsim B$ denotes the existence of an
absolute constant $C$ such that $A \ge B/C$.

We use $\|\cdot\|_2$ and $\|\cdot\|_\infty$ for the standard $\ell_2$
and $\ell_\infty$ norms. We use $B_2^m = \{x: \|x\|_2 \le 1\}$ to
denote the unit $\ell_2$ ball in $\R^m$. For two subsets $S, T \subset
\R^m$, the notation $S + T$ denotes the Minkowski sum, i.e.~the set
$\{s + t: s \in S, t \in T\}$. 

For a real-valued random variable $Z$ we use the notation $\|Z\|_p =
\left(\E |Z|^p\right)^{1/p}$.

\subsection{Concentrated Differential Privacy} \label{sec:cdp}

Recalling Definition~\ref{def:cdp}, we say that a randomized algorithm $\alg$ satisfies $\rho$-zCDP if for any two neighboring datasets $\db, \db'$ and all $\gamma \in (1, \infty)$, we have
\[\rendiv_\gamma(\alg(\db) \| \alg(\db')) \le \gamma \rho.\]
Here, $\rendiv_\gamma(\cdot \| \cdot)$ denotes the R\'{e}nyi divergence of order $\gamma$. For probability density functions $P, Q : \Omega \to \mathbb{R}$ with $P$ absolutely continuous with respect to $Q$, this quantity is defined as
\[\rendiv_\gamma(P \| Q) = \frac{1}{1-\gamma}\log\left(\int_{\Omega}
  P(x)^\gamma Q(x)^{1-\gamma} \ dx\right).\]
For two random variables $Y,Z$, the divergence $\rendiv_\gamma(Y\|Z)$
is defined as the divergence of their probability densities.

%\todo[inline]{Sasho: Above is $\alg(X)$ a random variable or a
%  density? In general do we want to say that sometimes
%  $\rendiv_\gamma$ is used for random variables, which is understood
%  to mean the densities?}
  
% While not all differentially private algorithmic techniques
%(e.g., those using the propose-test-release paradigm~\cite{DworkL09}) are compatible with CDP,
%it encompasses a rich subclass (arguably, the vast majority) of the
%techniques in the differential privacy literature. These techniques
%include Laplace and Gaussian noise addition, the exponential mechanism
%\cite{McSherryT07}, sparse vector~\cite{DR14-monograph}, and private
%multiplicative weights \cite{HardtR10}.  For several natural sets of statistical queries, techniques compatible with CDP 
%suffice to achieve optimal error with respect to all $(\varepsilon, \delta)$-differentially private algorithms~
%\cite{BunUV14}.

One of the crucial properties of CDP is the following tight
composition theorem which matches the guarantees of the so-called
``advanced composition'' theorem for approximate differential
privacy~\cite{DworkRV10}.
%
%As shown in~\cite{BunS16}, this formulation of concentrated differential privacy immediately yields a simple and optimal composition theorem.

\begin{lemma}[\cite{BunS16}]\label{lm:comp}
  Assume that the algorithm $\alg_1(\cdot)$ satisfies $\rho_1$-zCDP,
  and, for every $y$ in the range of $\alg_1$, the algorithm
  $\alg_2(\cdot, y)$ satisfies $\rho_2$-zCDP. Then the algorithm
  $\alg$ defined by $\alg(\db) = \alg_2(\db, \alg_1(\db))$ satisfies
  $(\rho_1 + \rho_2)$-zCDP. 
\end{lemma}

We remark that as a special case of Lemma~\ref{lm:comp}, one can take $\alg_2$ to be a $0$-zCDP algorithm which does not directly access the sensitive dataset $X$ at all. In this case, the combined algorithm $\alg$ satisfies $\rho_1$-zCDP, showing that zCDP algorithms can be \emph{postprocessed} without affecting their privacy guarantees.

Our algorithms are designed by carefully applying two basic building
blocks: the Projection and the Private Multiplicative Weights
mechanisms. Below we state their guarantees for the mean point
problem. 

In order to state the error guarantees for the projection mechanism,
we need a couple of definitions. First, let us define the \emph{support
function} of a set $S \subseteq \R^m$ on any $x \in \R^m$ by $h_S(x) =
\sup_{y \in S} \langle y, x\rangle$, where $\langle \cdot, \cdot
\rangle$ is the standard inner product. If $Z \sim N(0, I)$ is a
standard Gaussian random variable in $\R^m$, then we define the
\emph{Gaussian mean width} of a set $S \subseteq \R^m$ by $\gmw(S) = \E
h_S(Z)$. 

\begin{lemma}[\cite{NTZ}]\label{lm:proj}
  Let $m \in \mathbb{N}$ and let $\Delta > 0$. There exists a mechanism $\alg_{PM}$ (The Projection Mechanism) such
  that, for every finite set $\uni \subseteq \Delta\sqrt{m}\cdot B_2^m$, 
  \[
  \err^2(\uni, \alg, n) \le 
  \left(\frac{\Delta \gmw(\uni)}{n\sqrt{2\rho m}}\right)^{1/2}
  \le 
  \frac{\Delta (\log |\uni|)^{1/4}}{(2\rho)^{1/4}\sqrt{n}}.
  \]
  Moreover, $\alg_{PM}$ runs in time $\poly(|\uni|, n)$.
\end{lemma}

\todo[inline]{One reviewer complained that we did not define the
  mechanism. We do describe it on a high level in the intro and I do
  not think we need to do it again here. Or?}

%\todo[inline]{Mark: Why do we cite HLM for this?}
\begin{lemma}[\cite{HardtR10}]\label{lm:mwu}
  There exists a mechanism $\alg_{PMW}$ (The Private Multiplicative Weights Mechanism) such that, for any finite $\uni \subseteq
  \Delta\cdot [-1,1]^m$,
  \[
  \err^\infty(\uni, \alg, n) =
  O\left(\frac{\Delta (\log |\uni|)^{1/4}(\log m)^{1/2}}{\rho^{1/4}\sqrt{n}}\right).
  \]
  Moreover, $\alg_{PM}$ runs in time $\poly(|\uni|, n)$.
\end{lemma}

\subsection{Local Differential Privacy}
\label{sect:LDP}

In the local model, the private database $\db$ is distributed among
$n$ parties, each party holding exactly one element of $\db$. For
convenience, we index the parties by the integers from $1$ to $n$, and
denote by $\row_i$ the element of $\db$ held by party $i$. The parties
together with a server engage in a protocol $\Pi$ in order to compute
some function of the entire database $\db$. Here we consider
sequentially interactive protocols (with non-interactive ones as a
special case), as defined by Duchi, Wainright, and
Jordan~\cite{DJW-ASA}. The protocol is defined by a collection of
randomized algorithms $\Pi_1, \ldots \Pi_n$. Algorithm $\Pi_i$ takes
as input $\row_i$ and a message $Y_{i-1}$ received from party $i-1$,
and produces a pair $(Y_i, Z_i)$, where $Y_i$ is sent to party $i+1$,
and $Z_i$ is sent to the server. Parties $1$ and $n$ are exceptions:
$\Pi_1$ only takes $\row_1$ as input, and $\Pi_n$ only produces $Z_n$
as output. Then the server runs a randomized algorithm $\alg$ on inputs
$Z_1, \ldots, Z_n$ to produce the final output of the protocol.  We
use $\Pi(\db)$ to denote the union of all outputs of the
algorithms. The running time of the protocol is the total running time
of the algorithms $\Pi_1, \ldots, \Pi_n$ and $\alg$.

Note that a special case of a sequentially interactive protocol is a
non-interactive protocol, in which $\Pi_i$ ignores $Y_i$ and only
depends on its private input $\row_i$. Non-interactive protocols
roughly capture the randomized response model of
Warner~\cite{Warner-RR}, and their study in the context of
differential privacy goes back to~\cite{DworkMNS06}.

To formulate our privacy definition in the local model, let us recall
the notions of max-divergence and approximate max-divergence, defined
for any two random variables $X$ and $Y$ on the same probability space
by
\begin{align*}
  \rendiv_\infty(X \| Y) &= \sup_S \frac{\Pr(X \in S)}{\Pr(Y\in S)}
  &\rendiv_\infty^\delta(X \| Y) = \sup_S \frac{\Pr(X \in S)-\delta}{\Pr(Y\in S)},
\end{align*}
where the supremum is over measurable sets $S$ in the support
of $X$. With this notation, the standard definition of an $(\eps,
\delta)$-differentially private algorithm~\cite{DworkMNS06} is as
follows. 

\begin{defn}\label{defn:ADP}
  A randomized algorithm $\alg$ satisfies $(\eps,
  \delta)$-differential privacy if for
  datasets $\db$ and $\db'$ we have
  \[
  \rendiv^\delta_\infty(\alg(\db) \| \alg(\db')) \le \eps.
  \]
\end{defn}

We will need the simple composition theorem for differential
privacy. See the book~\cite{DR14-monograph} for a proof.
\begin{lemma}\label{lm:dp-comp}
  Assume that the algorithm $\alg_1(\cdot)$ satisfies $(\eps_1,
  \delta_1)$-differential privacy, and, for every $y$ in the range of
  $\alg_1$, the algorithm $\alg_2(\cdot, y)$ satisfies $(\eps_2,
  \delta_2)$-differential privacy. Then the algorithm $\alg$ defined
  by $\alg(\db) = \alg_2(\db, \alg_1(\db))$ satisfies $(\eps_1  +
  \eps_2, \delta_1 + \delta_2)$-differential privacy.
\end{lemma}

The privacy definition for a sequentially interactive protocol in the
local model we adopt is as follows.

\begin{defn}\label{defn:ldp}
  A protocol $\Pi$ in the local model satisfies $(\eps, \delta)$-local
  differential privacy (LDP) if the algorithm $\Pi_1$ satisfies $(\eps,
  \delta)$-differential privacy with respect to the single element
  dataset $\{\row_1\}$, and, for every $2 \le i \le n$ and every $y$
  in the range of $\Pi_{i-1}$,  the algorithm $\Pi_{i}(\cdot, y)$
  satisfies $(\eps, \delta)$-differential privacy with respect to the
  single element dataset $\{\row_i\}$.

  When a protocol satisfies $(\eps, 0)$-LDP, we also say that it
  satisfies $\eps$-LDP.
\end{defn}

\cut{We note that this privacy definition, proposed by Duchi et
al.~\cite{DJW-ASA}, is somewhat more restrictive than the two-party
differential privacy definition in~\cite{TwoDP} (which easily
generalizes to a multiparty setting) and the local differential
privacy definition in \cite{BNE08}. One difference is the restricted
interaction pattern. Another more subtle difference is that
differential privacy is required to hold for $\Pi_i$ no matter what
message  party $i-1$ sent to party $i$. This corresponds to
an adversarial setting in which parties may deviate from the protocol
to force other parties to reveal more information. By contrast, the
definitions in \cite{BNE08,TwoDP} are with respect to honest but
curious parties. }

We note that while our protocols work in the non-interactive pure LDP
model (i.e.~$\delta = 0$), our lower bounds work against the larger
class of
sequentially interactive protocols and approximate LDP
(i.e.~sufficiently small but nonzero $\delta$).
\cut{ It is an interesting
open problem to extend our lower bounds to the honest but curious
model of \cite{BNE08,TwoDP}.}

%\todo[inline]{Jarek: it would be useful to define max-divergence $D_\infty$ and state the pure local privacy in terms of max divergence.}

\subsection{Packings, Coverings, and Dudley's Inequality}

Recall the definitions of separation numbers given in the
Introduction: for a set $S \subseteq \mathbb{R}^m$  and a proximity
parameter $t > 0$ we denote
\begin{align*}
  \sep(S, t) &= \sup \{|T| : T \subseteq S, \quad \forall x, y \in T, x \ne y \implies \|x - y\|_2 > t\sqrt{m} \};\\
  \sep_\infty(S, t) &= \sup \{|T| : T \subseteq S, \quad \forall
x, y \in T, x \ne y \implies \|x - y\|_\infty > t\}.
\end{align*}
Note the non-standard scaling of $\sep(S,t)$, which we chose because
it corresponds better to the definition of average error. 

To prove the optimality of our algorithms, we make use of the well-known duality between packings (captured by separation numbers) and coverings.
We say that a set $T\subset \R^m$ is a \emph{$t$-covering} of $S$ with respect to metric $d$ if for every $x \in S$ there exists a point $y \in T$ such that $d(x, y) \le t$. This definition gives rise to the family of \emph{covering numbers}
%$\forall x\in S, \exists y\in T \text{ s.t. } d(x,y) \leq t$.
(in $\ell_2$ or $\ell_\infty$) of a compact set
$S \subseteq \R^m$, defined by
\begin{align*}
  \cov(S, t) &= \inf \{|T| : \forall x \in S \ \exists y \in T
  \text{ s.t. } \|x - y\| \le t\sqrt{m}\},\\
  \cov_\infty(S, t) &= \inf \{|T| : \forall x \in S \ \exists y \in T
  \text{ s.t. } \|x - y\|_\infty \le t\}.
\end{align*}

The next  lemma relating separation numbers to covering numbers is
folklore (see e.g.~Chapter 4 of \cite{AGM-book} for a proof).
\begin{lemma}\label{lm:packing-covering}
  Let $S$ be a compact subset of $\R^m$, and $t> 0$ be a real
  number. Let $T$ be a maximal subset of $S$ with respect to inclusion
  s.t.~$\forall x, y \in T: \|x-y\|_2 \ge t\sqrt{m}$
  (resp.~$\|x-y\|_\infty \ge t$). Then $T$ is a $t$-cover of $S$,
  i.e.~for any $x \in S$ there exists a $y \in T$ such that $\|x -
  y\|_2 \le t\sqrt{m}$ (resp.~$\|x-y\|_\infty \le t$). This implies
  \begin{equation*}
    \cov(S, t) \le \sep(S, t),\ \ \ \ \ \ \ 
    \cov_\infty(S, t) \le \sep_\infty(S, t).
  \end{equation*}
\end{lemma}

We will sometimes have to contend explicitly with the sets described
above. A set $T \subset \R^m$ is called a \emph{$t$-separated} set
with respect to a metric $d$ if for every $x, y \in T$, we have $d(x,
y) \geq t$. In what follows, when we discuss $t$-separated sets in the
context of the error measure $\err^2$ the underlying metric will be
the scaled $\ell_2$ norm $d(x,y) = \frac{1}{\sqrt{m}}\|x - y\|_2$, and in
the context of the error measure $\err^\infty$ the underlying metric
will be the $\ell_\infty$ norm $d(x,y) = \|x - y\|_\infty$.

Dudley's Inequality is a tool which allows us to relate the Gaussian mean width of a set $\uni \subset \R^n$ with the family of  covering numbers of $\uni$ at all scales. (Note that the normalization factor of $\sqrt{m}$ appearing in our definition of separation/covering numbers causes this statement to differ by a factor of $\sqrt{m}$ from its usual formulation.) %We are using convenient in our setting renormalization by $\sqrt{n}$ factor in the definition of separation/covering numbers, hence the statement differs by a factor of $\sqrt{n}$ from the standard formulation.
\begin{lemma}[{\cite[Chapter 11.1]{ledoux1991probability}}]
	\label{lm:dudley}
	For any subset $\uni \subset \R^m$, with diameter $\sqrt{m} \Delta(\uni)$, we have
	\begin{equation*}
		\gmw(\uni) \lesssim  \sqrt{m} \int_0^{\Delta(\uni)} \sqrt{\log \cov(\uni, t)} \ d t.
	\end{equation*}
\end{lemma}
%\todo[inline]{Mark: Should define the $\lesssim$ notation}

\subsection{Subgaussian Random Variables}

We recall the standard definition of a subgaussian random variable.
\begin{defn}
  We say that a mean zero random variable $Z\in \R^m$ is $\sigma$-subgaussian, if for every fixed $\theta \in \R^m$, we have
  \begin{equation*}
    \E \exp(\inprod{Z, \theta}) \leq \exp(\sigma^2 \|\theta\|^2).
  \end{equation*}
  
  For an arbitrary random variable $Z \in \R^m$ we say that it is $\sigma$-subgaussian if $Z- \E Z$ is $\sigma$-subgaussian.
\end{defn}
 
We recall some basic facts about subgaussian random variables in the appendix.

\section{Decompositions of the Universe}
In this section we show a simple decomposition lemma (Lemma~\ref{lm:decomposition}) that underlies all of our new algorithms. We begin by identifying an important property common to both error measures $\err^2$ and $\err^\infty$ which will be essential to Lemma~\ref{lm:decomposition}.
\begin{defn}[Subadditive error measure]
Let $\uni \subset \R^m$ be a finite universe enumerated as $\uni = \{
x_1, x_2, \ldots x_T \}$. For each $x_i$, consider an arbitrary
decomposition $x_i = x_i^1 + x_i^2$, and define $\uni^{(1)} = \{ x_i^1
~ : ~ 1 \le i \leq T\}$, $\uni^{(2)} = \{ x_i^2 ~ : ~ 1\le i \leq T \}$. This decomposition induces, for any dataset $X \in \uni^n$, a pair of datasets $X^1 \in (\uni^{(1)})^n, X^2 \in (\uni^{(2)})^n$. 

We say that an error measure $\err(X, \alg)$ is a
\emph{subadditive error measure} if for every finite universe
$\uni$, every decomposition as above, every dataset $X \in \uni^n$
and every pair of algorithms (or local protocols) $\alg^1, \alg^2$, we have
\begin{equation*}
  \err(X, \alg) \leq \err(X^1, \alg^1) + \err(X^2, \alg^2),
\end{equation*}
where the algorithm (resp.~local protocol) $\alg$ is defined by
$\alg(X) = \alg^1(X^1) + \alg^2(X^2)$.
\end{defn}

Both error measures of interest in this paper are subadditive error measures. 
\begin{claim}
    \label{clm:subadditive}
    Both $\err^2$ and $\err^\infty$ are subadditive error measures.
\end{claim}
The proof of this claim can be found in the appendix. It follows
directly from the triangle inequality for $\ell_2$ and $\ell_\infty$ norms
respectively.

\begin{lemma}
  \label{lm:decomposition}
  Let $\uni \subset \R^{m}$ be a subset of a Minkowski sum $\uni
  \subset \uni^{(1)} + \uni^{(2)} + \ldots +\uni^{(k)}$, and let
  $\pi_1, \ldots, \pi_k$ be functions, respectively, from $\uni$ to
  $\uni^{(i)}$ such that $x = \sum_{i = 1}^k{\pi_i(x)}$. Consider an
  arbitrary subadditive error measure $\err$.  Let $\alg^{1},
  \alg^{2}, \ldots \alg^k$ be a sequence of algorithms (respectively
  protocols in the local model) such that for every $j$,
    \begin{enumerate}
    \item $\err(\uni^{j}, \alg^{j}) \leq \alpha_j$, and 
	\item $\alg^j$ satisfies $\rho_j$-zCDP (resp.~$\eps_j$-LDP). 
          %\todo{Jarek: Can we somehow state this lemma in greater generality, to encompass LDP?}
    \end{enumerate}
    Then we can construct a $\rho$-zCDP mechanism (resp.~$\eps$-LDP
    protocol) $\alg$ with $\err(\uni,
    \alg) \leq \sum_{j=1}^k \alpha_j$, where $\rho = \sum_{j=1}^k \rho_j$
    (resp.~$\eps = \sum_{j = 1}^k{\eps_j}$).
    
    Moreover, the running time of $\alg$ is bounded by the sum of the
    running times of $\alg^1, \ldots, \alg_k$, and the sum of the
    running times to compute $\pi_1, \ldots, \pi_k$ on $n$ vectors
    from $\uni$. If $\alg^1,
    \ldots, \alg^k$ are non-interactive local protocols, then so is
    $\alg$. 
\end{lemma}
\begin{proof}
  We first prove the lemma for CDP.
  For a database $\db \in \uni^n$, we can consider a sequence of
  induced databases $\db^1, \ldots, \db^k$, where $\db^j$ is derived
  from $\db$ by applying $\pi_j$ pointwise to each one of its
  elements. Given a database $\db$ we compute independently
  $\alg^{j}(\db^{j})$ for every $j$, and release $\sum_{j=1}^k
  \alg^{j}(\db^{j})$.

  The privacy of  $\alg$ follows from the composition properties of zCDP (Lemma~\ref{lm:comp}), and postprocessing --- i.e. by the composition lemma we know that releasing $(\alg^{1}(X^{1}), \ldots \alg^{k}(X^k))$ satisfies $\rho$-zCDP, and by postprocessing $\sum_{j=1}^k \alg^{j}(X^j)$ has the same privacy guarantee.
	
    Moreover, the error bound is satisfied by inductively applying subadditivity of the error measure. Indeed, for any specific database $\db$, we have $\err(\db, \alg) \leq \sum_{j=1}^k \err(\db^j, \alg^j)$, and therefore
    \begin{equation*}
        \err(\uni, \alg) = \max_{\db \in \uni^n} \err(\db, \alg) \leq \sum_{j =1}^k \max_{\db^j \in (\uni^{(j)})^n} \err(\db^j, \alg^j) \leq \sum_{j =1}^k \err(\uni^{j}, \alg^i).
    \end{equation*}
    The proof for LDP is analogous: each party $i$, given input
    $\row_i$, for each $j$, $1 \le j \le k$ runs the local protocol
    $\alg^{j}$ with input $\pi_j(\row_i$). The protocols can be run
    in parallel. At the end the server can compute and output the sum
    of the outputs of the local protocols. The error analysis is the
    same as above, and the privacy bound $\varepsilon \leq \sum_{j=1}^k
    \varepsilon_j$ follows from the simple composition theorem
    (Lemma~\ref{lm:dp-comp}) for (pure) differential privacy.
\end{proof}

\section{Algorithms for Concentrated Differential Privacy}

In this section we define our two new algorithms in the centralized
model. In the subsequent section we describe how to adapt them to the
local model. 

\subsection{The Coarse Projection Mechanism}
In this section, we prove Theorem~\ref{thm:ub-coarse} giving the guarantees of the coarse projection mechanism.

For a finite $\uni \subset \R^m$, let $\uni^{(1)} \subset \R^m$ be an
inclusion-maximal $\frac{\alpha}{2}$-separated subset of $\uni$ with
respect to the metric $d(x,y) = \frac{1}{\sqrt{m}}\|x - y\|_2$. Let $\uni^{(2)} = \frac{\alpha}{2}\sqrt{m} B_2^m$. We claim that $\uni \subset \uni^{(1)} + \uni^{(2)}$: this follows since, by Lemma~\ref{lm:packing-covering}, $\uni^{(1)}$ is a $\frac{\alpha}{2}$-cover of $\uni$. 
%To see this, observe that for every $s \in \uni$, there is some $\tilde{s} \in \uni^{(1)}$ with $\|\tilde{s} - s\| \leq \alpha \sqrt{m}$. This holds by the maximality of $\uni^{(1)}$; for if it did not hold, we could enlarge $\uni^{(1)}$ to an $\alpha$-separated set $\uni^{(1)} \cup \{s\}$.

Let $\alg^1 : (\uni^{(1)})^n \to \uni^{(1)}$ be the projection mechanism (as in Lemma~\ref{lm:proj}) and let $\alg^2$ be the trivial $0$-zCDP mechanism where $\alg^2(x) = 0$ for all $x$.  Note that $\alg^2$ has error $\err^2(\uni^{(2)}, \alg^2) \leq \frac{\alpha}{2}$.

We now invoke Lemma~\ref{lm:decomposition} with $\uni^{(1)},
\uni^{(2)}$ and $\alg^1, \alg^2$ as described, and using the
(subadditive) error measure $\err^2$. As mentioned in the
introduction, this gives  the following simple mechanism:
\begin{enumerate}
\item Round each element $\row$ of the dataset $\db$ to the nearest
  point $\row^{(1)}$ in the covering $\uni^{(1)}$ to get a rounded dataset
  $\db^{(1)}$.
\item Add enough Gaussian noise to $\bar{\db}^{(1)}$ to preserve
  $\rho$-zCDP; let the resulting  noisy vector be $\tilde{Y}$.
\item Output the closest point in $\conv\{\db^{(1)}\}$ to $\tilde{Y}$
  in $\ell_2$. 
\end{enumerate}

 We will use the following lemma to analyze the error incurred by the
 projection mechanism (corresponding to steps 2.~and 3.~above). 
\begin{lemma}
	\label{lm:coarse-dudley}
	For any $\alpha > 0$ and any $\frac{\alpha}{2}$-separated set $S \subset \R^m$ with diameter $\sqrt{m} \Delta$, we have

	\begin{equation*}
		g(S) \lesssim \sqrt{m} \log(\Delta/\alpha) \sup \left\{ t \sqrt{\log \sep(S, t)}\ :\ t \geq \alpha/4 \right\}.
	\end{equation*}
\end{lemma}
\begin{proof}%[Proof of Lemma~\ref{lm:coarse-dudley}]
	By Dudley's inequality (Lemma~\ref{lm:dudley}), we have 
	\begin{align}
		g(S) & \lesssim \sqrt{m} \int_0^\Delta \sqrt{\log \cov(S, t)} \d t \nonumber \\
		& \leq \int_0^{\alpha/4} \sqrt{\log \cov(S, t)} \d t + \int_{\alpha/4}^{\Delta} \sqrt{\log \cov(S, t)}.  \label{eq:two-terms}
	\end{align}

	Now we can bound the two summands on the right hand side
        separately. Note that for $t < \alpha/4$, because $S$ is
        $\frac{\alpha}{2}$-separated, we have $\cov(S, t) = |S| \le \sep(S, \alpha/2)$ --- indeed, every covering of radius $\frac{\alpha}{4}$ has to contain every point of $S$. Therefore
	\begin{equation}
		\int_0^{\alpha/4} \sqrt{\log \cov(S, t)} \d t \leq \frac{\alpha}{4} \sqrt{\log \sep(S, \alpha/4)}. \label{eq:first-term}
	\end{equation}

	On the other hand, for the second summand, we have
	\begin{align}
		\int_{\alpha/4}^{\Delta} \sqrt{\log \cov(S, t)} \d t & =
		\int_{\alpha/4}^{\Delta} t \sqrt{\log \cov(S, t)} \frac{\d t}{t} \nonumber \\
		& \leq \sup_{t > \alpha/4}\{t \sqrt{\log \cov(S, t)}\} \int_{\alpha/4}^{\Delta} \frac{\d t}{t} \nonumber \\
		& \leq \sup_{t > \alpha/4}\{t \sqrt{\log \cov(S, t)}\}
                \ln\frac{\Delta}{\alpha/4} \nonumber\\
                & \leq \sup_{t > \alpha/4}\{t \sqrt{\log \sep(S, t)}\}
                \ln\frac{\Delta}{\alpha/4}. \label{eq:second-term}
	\end{align}
	The last inequality follows from the duality between packing
        and covering numbers (Lemma~\ref{lm:packing-covering}). 
	By~\eqref{eq:first-term}~and~\eqref{eq:second-term}, we have a bound
	\begin{align*}
          g(S) & \lesssim \sqrt{m} \log \frac{\Delta}{\alpha} \sup_{t \geq \alpha/4}( t \sqrt{\log \sep(S, t)}).
	\end{align*}
        This completes the proof. 
\end{proof}

%The proof of this lemma is based on Dudley's Inequality, and appears in the appendix.
%By applying Lemma~\ref{lm:decomposition} with mechanism $\alg^{(1)}$ --- projection mechanism as in Lemma~\ref{lm:proj}, and trivial mechnaims $\alg^{(2)}$, we deduce
Using the guarantee of the projection mechanism (Lemma~\ref{lm:proj}), Lemma~\ref{lm:decomposition} shows that
\begin{align*}
	\err^2(\uni, \alg) & \leq \err^2(\uni^{(1)}, \alg^{1}) + \err^2(\uni^{(2)}, \alg^{2})  \\
    & \leq  \left(\frac{\Delta \gmw(\uni^{(1)})}{ n\sqrt{2\rho m}}\right)^{1/2} + \frac{\alpha}{2}.
\end{align*}
%\todo[inline]{Mark: Is it 2 in the denominator, or $\sqrt{2}$?}
Hence, as soon as $n \gtrsim \frac{\Delta \gmw(\uni^{(1)})}{\sqrt{\rho m} \alpha^2}$, the first term is bounded by $\alpha/2$, and the total error $\err^2(\alg, \uni) \leq \alpha$. By applying Lemma~\ref{lm:coarse-dudley}, we can deduce that it is enough to have
\begin{align*}
    n & \gtrsim \frac{\Delta}{\sqrt{\rho m} \alpha^2}  \sqrt{m} \log(\Delta/\alpha) \sup \left\{ t \sqrt{\log \sep(\uni^{(1)}, t)}\ :\ t \geq \alpha/4 \right\} \\
    & = \frac{\Delta \log(\Delta/\alpha)}{\sqrt{\rho} \alpha^2} \sup \left\{ t \sqrt{\log \sep(\uni^{(1)}, t)}\ :\ t \geq \alpha/4 \right\}
\end{align*}
Since $\uni^{(1)} \subset \uni$, we have $\sep(\uni^{(1)}, t) \leq \sep(\uni, t)$ for each $t$. Therefore, as soon as
\begin{equation*}
	n \gtrsim \frac{\Delta \log(\Delta/\alpha)}{\sqrt{\rho} \alpha^2} \sup \left\{ t \sqrt{\log \sep(\uni, t)}\ :\ t \geq \alpha/4 \right\}
\end{equation*}
we have $\err^2(\alg, \uni) \leq \alpha$. Finally, because $\uni \subset [0, 1]^m$ in the statement of Theorem~\ref{thm:ub-coarse} has diameter bounded by $\sqrt{m}$, we can take $\Delta = 1$. This concludes the proof of Theorem~\ref{thm:ub-coarse}.

\subsection{The Chaining Mechanism for Average Error \label{sec:cdp-chain}}

In this section we will prove Theorem~\ref{thm:ub-chain}. The main ingredient of the proof is following decomposition lemma for a finite set $\uni \subset [0,1]^m \subset \sqrt{m} B_2$. This decomposition will be used together with Lemma~\ref{lm:decomposition}, to get the desired algorithm.

\begin{lemma}
  \label{lm:chaining-decomposition}
  Let $\| \cdot \|$ be a norm on $\R^m$ and let $B = \{x \in \R^m: \|x\| \le 1\}$ be its unit ball.
For $\Delta \ge 0$, $\alpha < 1$, and an arbitrary set $\uni \subset \Delta B$, there
exist a sequence of subsets $\uni^{(1)}, \ldots \uni^{(k)}$, with $k =
\lceil \log \frac{2}{\alpha}\rceil $, such that $\uni \subset
\uni^{(1)} + \ldots + \uni^{(k)} + \frac{\alpha}{2} \Delta B$, where
$\uni^{(i)} \subset 2^{-i + 1} \Delta B$, and $\uni^{(i)}$ is
$2^{-i}$-separated with respect to the metric $d(x,y) = \|x -
y\|/\Delta$. 
Moreover, $\uni^{(1)}, \ldots \uni^{(k)}$ can be computed in time
polynomial in $m$ and $|\uni|$ if $\|\cdot\|$ can be computed in time
polynomial in $m$.
\end{lemma}
\begin{proof}
  Let $S_1, S_2, \ldots S_k \subset \uni$ be a sequence of subsets
  such that each $S_j$ is a maximal $2^{-j}$-separated set in the
  metric $d(x,y)$ defined above. By the maximality of each set, we
  have
\begin{equation}
	\forall x \in \uni \quad \exists y \in S_j \quad \text{s.t.} \quad  \|x - y\| \leq 2^{-j} \Delta.
\end{equation}

Let us define a mapping $\pi_j :\uni \to S_j$ such that for each $x$, we have
\begin{equation} \label{eqn:proj-length}
\|\pi_j(x) - x\| \leq 2^{-j} \Delta
\end{equation}
 Finally, we define $\uni^{(1)} = S_1$, and for $j > 1$, we will take $\uni^{(j)} := \{ x - \pi_{j-1}(x) \ :\ x \in S_{j}\}$.

We first observe that $\uni \subset \uni^{(1)} + \uni^{(2)} + \ldots + \uni^{(k)} + 2^{-k} \Delta  B$. Indeed, for any $x \in \uni$, we can take $x_k = \pi_k(x)$, and $x_j = \pi_j(x_{j+1})$ for every $j < k$. Then for every $j< k$ we have $x_{j+1} - x_{j} \in \uni^{(j+1)}$, and we have the telescoping decomposition
\begin{equation*}
	x = x_1 + \sum_{j=1}^{k-1} (x_{j+1} - x_j) + (x - x_k).
\end{equation*}
Finally, by definition of $x_k$, we have $\|x - x_k\| = \|x - \pi_k(x)\| \leq 2^{-k}  \Delta$, or equivalently $x - x_k \in 2^{-k}  \Delta B$.

Observe now that every set $\uni^{(j)}$ satisfies $\uni^{(j)} \subset
2^{-j + 1} \Delta B$. This is the case because every element of
$\uni^{(j)}$ takes the form $x - \pi_{j-1}(x)$ for some $x \in \uni$, and all
of these vectors are bounded by $2^{-j+1} \Delta$ in the $\|\cdot\|$
norm by \eqref{eqn:proj-length}.  
%On the other hand,
%by the definition of $\uni^{(j)}$, we have $|\uni^{(j)}| \leq |S_j|
%\leq \sep(\uni, 2^{-j})$. 
%Moreover, directly by construction $|\uni^{(i)}| \leq |S_i| = \sep(\uni, 2^{-i})$.
\end{proof}

Using Lemma~\ref{lm:chaining-decomposition} with the Euclidean norm
$\|\cdot\| = \|\cdot\|_2$ and $\Delta = \sqrt{m}$, we get the decomposition $\uni
\subset \uni^{(1)} + \ldots + \uni^{(k)} + \frac{\alpha}{2}\sqrt{m}
B_2^m$, to which we apply Lemma~\ref{lm:decomposition}. Each algorithm
$\alg^j$ for $1 \le j \leq k$ is the projection mechanism (as in
Lemma~\ref{lm:proj}) instantiated with privacy parameter $\rho_j =
\rho/k$, and $\alg^{k+1}$ (corresponding to the set $\uni^{(k+1)} =
\frac{\alpha}{2}\sqrt{m} B_2^m$) is the trivial mechanism which always
outputs $0$.

Let us now analyze the composed mechanism guaranteed by Lemma~\ref{lm:decomposition}. First, observe that this mechanism satisfies $\rho$-zCDP. Indeed, as guaranteed by Lemma~\ref{lm:decomposition}, it satisfies zCDP with privacy parameter $\sum_{j=1}^{k+1} \rho_j = \sum_{j=1}^{k} \frac{\rho}{k} + 0 = \rho$.

We now turn to analysis of the error of the composed mechanism. Observe first that the error of the trivial algorithm $\alg^{k+1}$ is  bounded as $\err^2(\alg^{k+1}, \frac{\alpha}{2} \sqrt{m}B_2^m) \leq \frac{\alpha}{2} $. 
Let us now consider the error of the mechanism $\alg^j$, for $1 \le j
\leq k$. The set $\uni^{(j)}$ is contained in $2^{-j + 1} \sqrt{m}
B_2$, and $|\uni^{(j)}| \leq \sep(\uni, 2^{-j})$ because $\uni^{(j)}$
is $2^{-j}$-separated --- hence we can apply Lemma~\ref{lm:proj} to deduce that for $1 \le j \leq k$, we have
\begin{equation*}
    \err^2(\alg^j, \uni^{(j)}) \leq \frac{2^{-j + 1}  (\log \sep(\uni, 2^{-j} ))^{1/4}}{(2\rho/k)^{1/4} \sqrt{n}}.
\end{equation*}
We wish to pick $n$ large enough, so that for every $j \leq k$, we have $\err^2(\alg^j, \uni^{(j)}) \leq \frac{\alpha}{2k}$. After rearranging, this will be satisfied as long as for every $j$, we have
\begin{equation*}
	n \gtrsim \frac{2^{- 2 j} \sqrt{\log \sep(\uni, 2^{-j})} k^{5/2}}{\sqrt{\rho} \alpha^2}
\end{equation*}
and since $\sup_{j \leq k} 2^{-2j} \sqrt{\log \sep(\uni, 2^{-j})} \leq \sup_{t \geq \alpha/2} t^2 \sqrt{\log \sep(\uni, t)}$, it is enough to have
\begin{equation*}
	n \gtrsim \frac{\sup \{ t^2 \sqrt{\sep(\uni, t)} \ :\ t \geq \alpha/2\} k^{5/2}}{\sqrt{\rho} \alpha^2}
\end{equation*}
Finally, recalling that $k = O(\log \frac{1}{\alpha})$, and the statement of Theorem~\ref{thm:ub-chain} follows.
\subsection{The Chaining Mechanism for Worst-Case Error}
The proof of Theorem~\ref{thm:ub-infty} is analogous to the previous proof of Theorem~\ref{thm:ub-chain}. Here we only give a brief overview of the proof, and defer the details to the appendix.

Our algorithm for worst-case error again applies
Lemma~\ref{lm:decomposition} to a suitable decomposition $\uni \subset
\uni^{(1)} + \cdots + \uni^{(k)} + \frac{\alpha}{2} B_\infty^m$, which
we get from Lemma~\ref{lm:chaining-decomposition} with the norm
$\|\cdot \| = \|\cdot\|_\infty$ and $\Delta = 1$. 
The basic algorithm we use as a building block to solve each subproblem on universe $\uni^{(j)}$ is the Private Multiplicative Weights algorithm (Lemma~\ref{lm:mwu}) instead of the projection mechanism.

\section{Algorithms for Local Differential Privacy}

In this section we describe our results in the local model.  Once
again, we consider the mean point problem, where each party receives
$x_i \in \uni \subset \R^m$, which together form a dataset $\db$, and
the goal is to approximately and privately compute $\bar{\db} =
\frac{1}{n} \sum_{i=1}^n \row_i$. We show how to adapt the coarse projection and
the chaining mechanisms to the local model and via packing lower
bounds show that the sample complexity of our mechanisms is optimal up
to constants when the error parameter $\alpha$ is constant.

\subsection{Local Projection Mechanism}

First we discuss a locally private analogue of the projection mechanism. Variants
of it appear to have been known to experts, and, independently of our
work, Bassily recently published and analyzed a slightly different
version~\cite{Bassily18}. The main observation behind
the local projection mechanism is that Gaussian noise addition in the
usual projection mechanism can be replaced with a locally private
point release mechanism, while the projection step can be performed by
the server. 
The error guarantee of the
 mechanism --- to be described in full shortly --- is captured
by the following theorem.

\begin{theorem}\label{thm:local-pm}
  There exist a non-interactive $\varepsilon$-LDP protocol for the mean point problem
  $\Pi_{LPM}$, such that for every finite set $\uni \subseteq \Delta
  \sqrt{m} B_2^m$,
  \begin{equation*}
    \err^2(\uni, \Pi_{LPM}, n) \lesssim \left(\frac{\Delta \gmw(\uni)}{\varepsilon \sqrt{n m}}\right)^{1/2} \leq \left(\frac{\Delta^4 \log |\uni|}{\varepsilon^2 n}\right)^{1/4}.
  \end{equation*}
  Moreover, $\Pi_{LPM}$ has running time polynomial in $n$, $m$, and
  $|\uni|$. 
\end{theorem}

This error bound should be compared with the bounds for concentrated
differential privacy in the centralized model (Lemma~\ref{lm:proj}) --- the
dependence on the size of the database and the privacy parameter is
worse, as the error scales down as $n^{-1/4}$ with the size of the
database in local model, as opposed to $n^{-1/2}$ in centralized model.

We will now proceed with a description of $\Pi_{LPM}$ and the proof of
this theorem. On input $\row_i$, each party $i$ samples a random
variable $W_{\row_i}$ and sends it to the server. In
Section~\ref{sec:local-release} we describe
the distribution of $W_\row$ and prove the following properties of it.
\begin{lemma}
  \label{lem:local-release}
  For any $x \in \Delta \sqrt{m} B_2^m$ there exists a random variable
  $W_x$ such that 
  \begin{enumerate}
  \item $W_x$ can be sampled in time polynomial in $m$ on input $x$,
  \item $\E W_x = x$,
  \item $W_x$ is $\sigma$-subgaussian with $\sigma =
    O(\varepsilon^{-1}\Delta \sqrt{m})$, and
  \item for any $x, y \in B_2^m$, we have $\rendiv_\infty(W_x\, || \, W_y) \leq \varepsilon$.
  \end{enumerate}
\end{lemma}

Assuming this lemma, we will proceed with the proof of
Theorem~\ref{thm:local-pm}. The protocol is as follows:
\begin{enumerate}
\item Each party $i$ samples $W_{\row_i}$ and sends it to the server; 
\item The server computes the average $\bar{W} := \frac{1}{n} \sum
  W_{x_i}$ of the responses. 
\item The server computes a projection of $\bar{W}$ onto the convex hull
  of $\uni$, that is
\begin{equation}
		W^* = \argmin_{w \in \mathrm{conv}(\uni)} \|w - \bar{W}\|_2,
\end{equation}
and outputs $W^*$.
\end{enumerate}

Note that this mechanism satisfies $\varepsilon$-LDP --- the local
differential privacy property corresponds directly to the bound on the
max-divergence of the variables $W_x$ in item 3 in the statement of
Lemma~\ref{lem:local-release}.

The accuracy guarantee follows from the original accuracy analysis of
the projection mechanism --- the following statement is proved, for
instance
in~\cite{DBLP:journals/corr/abs-1212-0297,Dwork:2014:UCR:2582112.2582123}. We give the short proof in the appendix for completeness. 
%for the Gaussian case, but the proof generalizes to subgaussian
%without any modification. We provide the proof in the appendix for
%completeness. \todo{Should we?}
%\todo[inline]{I changed things a bit, and now I reduce the subgaussian case to  the Gaussian case via majorizing measures. I think no proof is now necessary.}
\begin{lemma}
  \label{lm:projection-error}
  Consider a convex polytope $K$, and random variable $W$ with $\E W = x \in K$. Take $W^* := \argmin_{w \in K} \|w - W\|_2$. Then
  \begin{equation*}
    \E \|W^* - x\|_2^2 \lesssim \E \sup_{w \in K}{\inprod{w, W - x}}.
  \end{equation*}
\end{lemma}
Notice that if $W$ is a Gaussian centered at $x$, then the right hand side above is
$\gmw(K)$. Lemma~\ref{lm:subgaussian-mw} then implies that 
if $W$ is $\sigma$-subgaussian, then the right hand side above is
bounded, up to constants, by $\sigma\gmw(K)$. We have the following
corollary.

\begin{corr}\label{corr:projection-error}
  Consider a convex polytope $K$, and $\sigma$-subgaussian random variable $W$ with $\E W = x \in K$. Take $W^* := \argmin_{w \in K} \|w - W\|_2$. Then
  \begin{equation*}
    \E \|W^* - x\|_2^2 \lesssim  \sigma\gmw(K).
  \end{equation*}
\end{corr}

Note that in $\Pi_{LPM}$, as described above, $\E \bar{W} = \bar{\db}$,
and moreover $\bar{W}$ is $\sigma$-subgaussian with $\sigma = O(\frac{\Delta \sqrt{m}}{\sqrt{n} \varepsilon})$. By applying
Corollary~\ref{corr:projection-error} with this $\sigma$, and $K = \mathrm{conv}(\uni)$,
and using our definition of average error
\begin{equation*}
  \err^2(\uni, \Pi_{LPM}, n) = \left(\frac{1}{m} \E \|W^* - \bar{\db}\|_2^2\right)^{1/2} \leq \left(\frac{\Delta \gmw(\uni)}{\sqrt{\varepsilon n m}}\right)^{1/2} 
\end{equation*}
we complete the proof of Theorem~\ref{thm:local-pm}.

\subsubsection{Local point release \label{sec:local-release}}

In this section we prove Lemma~\ref{lem:local-release}. We focus on
the case $x \in B_2^m$, i.e.~$\Delta = \frac{1}{\sqrt{m}}$ in the
statement of the lemma. The general case follows by scaling $W_x$
defined below by $\Delta \sqrt{m}$. 

Our construction of $W_x$ is a variant of the distribution defined in
\cite{DJW-ASA} for the mean estimation problem in the local model. We
replace their use of the uniform distribution on the sphere with the
Gaussian, which is slightly easier to analyze. The definition of $W_x$
follows.
\begin{defn}
  \label{def:v_x}
  For a given $x \in B_2^m$, and $\varepsilon > 0$, we define a random variable $W_x$ as follows.
  \begin{itemize}
  \item Take $U_x \in \{\pm \frac{x}{\|x\|_2}\}$, such that $\E U_x = x$.
  \item Draw $Z$ at random from a Gaussian distribution $Z \sim \mathcal{N}(0, I)$.
  \item Given $Z$ and $U_x$, take $S_x \in \{\pm 1\}$ at random, from a distribution such that $\E [ S_x | Z, U_x ] = \frac{\eps}{3} \sign(\inprod{Z, U_x})$.
  \item Define $W_x =\frac{3}{\sqrt{\pi} \varepsilon} Z S_x $.
  \end{itemize}
\end{defn}

We will first show that the expectation of $W_x$ is correct --- this
is an almost direct consequence of the construction.

\begin{claim}
  For any $x \in B_2^m$ and the random variable $W_x$ defined as in Definition~\ref{def:v_x}, we have $\E W_x = x$.
\end{claim}
\begin{proof}
If we condition on $U_x = y_0$, we have 
\begin{equation*}
  \E[\inprod{W_x, y_0}\, |\, U_x = y_0] = \E [ \inprod{Z, y_0} \cdot \frac{\eps}{3} \sign(\inprod{Z, y_0}) \cdot \frac{3}{\sqrt{\pi} \varepsilon}] = \E \frac{|\inprod{Z, y_0}|}{\sqrt{\pi}} = \|y_0\|_2 = 1.
\end{equation*}
On the other hand, for any $y_1 \perp y_0$, we have 
\begin{equation*}
  \E [\inprod{W_x, y_1} | U_x = y_0] = 
  \frac{1}{\sqrt{\pi}} \E  \sign(\inprod{Z, y_0}) \inprod{Z, y_1} =
  \frac{1}{\sqrt{\pi}} \E[ \sign(\inprod{Z, y_0})] \E[\inprod{Z, y_1}]
  = 0,
\end{equation*}
because $\inprod{Z, y_0}$ and $\inprod{Z, y_1}$ are independent. Therefore $\E[W_x | U_x = y_0] = y_0$, and, by the law of total expectation, $\E W_x = \E_{U_x}[ \E_{W_x} [ W_x | U_x ] ] = \E U_x = x$.
\end{proof}

By a direct application of Fact~\ref{fact:subgaussian-by-sign} to the
variables $Z$ and $S_x$, we obtain the following claim.
\begin{claim}
  For any $x \in B_2^n$, the random variable $W_x$ is $O(\varepsilon^{-1})$-subgaussian.
\end{claim}

We will shift focus to bounding the max divergence between $W_x$ and $W_y$. A direct calculation shows the following well-known fact.

\begin{fact}
  \label{fact:randomized-response}
  For any $0 \leq \eps \le \frac12$ and any two random variables $S_1, S_2 \in \{\pm 1\}$ with $|\E S_1|,|\E S_2| \leq \varepsilon$, we have $\rendiv_\infty(S_1\, ||\, S_2) \leq 3\varepsilon$.
\end{fact}

\begin{claim}
  For $x,y \in B_2^n$, we have $\rendiv_\infty(W_x\, ||\, W_y) \leq \varepsilon$.
\end{claim}
\begin{proof}
  We can consider explicitly releasing $(Z,  S_x)$ instead, since they entirely define $W_x$, and so, by the data-processing inequality,  $\rendiv_\infty(W_x || W_ y) \leq \rendiv_\infty( (Z, S_x) || (Z, S_y))$. For any fixed $z$, the max-divergence between $(S_x | Z = z)$ and $(S_y | Z = z)$ is bounded by $\eps$ by Fact~\ref{fact:randomized-response} --- this is enough to conclude that densities of $(Z, S_x)$ and $(Z, S_y)$ are within a $e^{\varepsilon}$ multiplicative factor.
\end{proof}

\subsection{Local Coarse and Chaining Mechanisms}

To achieve Theorem~\ref{thm:ldp} the local projection mechanism can be lifted to a coarse projection and chaining mechanism via the same general framework as discussed in Section~\ref{sec:cdp-chain}.

The proof of the first bound for $\alg_1$ in Theorem~\ref{thm:ldp} is
essentially identical to Theorem~\ref{thm:ub-coarse}: we consider the
same decomposition of the universe $\uni \subset \uni^{(1)} +
\uni^{(2)}$, where $\uni^{(1)}$ is a maximal $\alpha/2$ separated set,
and $\uni^{(2)} := \frac{\alpha}{2} \sqrt{m} B_2$. We apply
Lemma~\ref{lm:decomposition}, using a local projection mechanism as
the basic algorithm $\alg^{(1)}$, and the trivial mechanism that
always outputs $0$ as the algorithm $\alg^{(2)}$. Theorem~\ref{thm:local-pm} bounds the error of this protocol in terms of $g(\uni^{(1)})$, and Lemma~\ref{lm:coarse-dudley} translates this to a bound  in terms of covering numbers of $\uni$ as in the theorem statement.  The total error of the composed mechanism is bounded as
\begin{equation*}
		\err^2(\uni, \alg_1, n) \leq \left(\frac{(\log \alpha^{-1})^2 \sup \left\{ t^2 \cdot \log(\sep(\uni, t))\right\}}{\varepsilon^2 n}\right)^{1/4} + \frac{\alpha}{2},
\end{equation*}
therefore if we pick $n$ as in \eqref{eq:ldp-coarse-ub}, we have $\err^2(\uni, \alg_1, n) \leq \alpha$ as desired.

For the construction of the chaining algorithm $\alg_2$ we use instead
the decomposition of the set $\uni$ guaranteed by
Lemma~\ref{lm:chaining-decomposition}. We apply
Lemma~\ref{lm:decomposition} using a local projection mechanism as the
basic algorithms $\alg^{(1)}, \ldots \alg^{(k)}$ instead of the projection mechanism as it was the case in Theorem~\ref{thm:ub-chain} --- the dependence of the error of the local projection mechanism on the sample size is quadratically worse, which yields quadratic loss in the sample complexity of the final chaining mechanism, and we lose an additional $\log(1/\alpha)$ term because of the weaker composition theorem for pure differential privacy.

\section*{Acknowledgements}

Jaros\l{}aw B\l{}asiok is supported by ONR grant N00014-15-1-2388.
Aleksandar Nikolov is supported by an NSERC Discovery Grant (RGPIN-2016-06333).

\bibliographystyle{alpha}
\bibliography{biblio}

\appendix
\section{Missing proofs}
\begin{proof}[Proof of Claim~\ref{clm:subadditive}]
	Consider the $\ell_2$ (semi)norm $\|\cdot\|_{L_2(\ell_2)}$ defined on random variables $Z \in \R^m$ by $\|Z\|_{L_2(\ell_2)} = (\E \|Z\|_2^2)^{1/2}$. By the triangle inequality for this norm, we have
	\begin{align*}
          \err^2(\db, \alg) & = \left(\E \|\alg(\db) - \bar{\db}\|_2^2\right)^{1/2} \\
        & = \left( \E \| \alg(\db^1) + \alg(\db^2) - \bar{\db}^1 - \bar{\db}^2 \|_2^2\right)^{1/2} \\ 
        & \leq \left( \E \| \alg(\db^1) - \bar{\db}^1 \|_2^2\right)^{1/2} + \left(\E \|\alg(\db^1) - \bar{\db}^2\|_2^2\right)^{1/2} \\
        & = \err^2(\db^1, \alg^1) + \err^2(\db^2, \alg^2)
	\end{align*}

    Similarly, $\err^\infty$ is subadditive by the triangle inequality for (semi)norm $\|Z\|_{L_1(\ell_\infty)} = \E \|Z\|_\infty$ for $\R^m$-valued random variables.
\end{proof}

\begin{proof}[Proof of Theorem~\ref{thm:ub-infty}]
%     For a given universe $\uni \subset [0,1]^m$, consider a sequence of sets $S_1, S_2, \ldots S_k$, where $k$ will be specified later, such that each $S_j$ is maximal $2^{-j}$-separated set in $\uni$, with respect to $\ell_\infty$ metric. Note that by maximality of $S_j$ for each $x \in \uni$, and every $j$, there is some $y \in S^{j}$ such that $\|x - y\|_\infty \leq 2^{-j}$. Let us call the appropriate selection function $\pi_j :\uni \to S_j$, that is $\forall x\in \uni \forall j\in [k],~\|x - \pi_j(x)\|_\infty \leq 2^{-j}$.

%     Let us consider sequence of sets, where $\uni^{(1)} = S_1$, and $\uni^{(j)} = \{ x - \pi_{j-1}(x) : x \in S^{j} \}$ for $1 < j \leq k$. We claim that $\uni \subset \uni^{(1)} + \uni^{(2)} + 2^{-k} B_\infty^m$. Indeed, for any $x \in \uni$, we can consider a sequence $x_k = \pi_k(x)$, and $x_{j-1} = \pi_{j-1}(x_j)$ for $1 < j \leq k$. We have telescoping decomposition
%     \begin{equation*}
%         x = x_1 + \sum_{j=2}^{k} (x_{j} - x_{j-1}) + (x - x_k)
%     \end{equation*}

%    Clearly, by definition $x_1 \in S_1 = \uni^{(1)}$, for larger $j$
%    we have $x_{j} - x_{j-1} = x_j - \pi_{j-1}(x_j) \in \uni^{(j)}$,
%    and finally $\|x - x_k\|_\infty = \|x - \pi_k(x)\|_\infty \leq
%    2^{-k}$. Moreover $|\uni^{(j)}| \leq \sep_\infty(\uni, 2^{-j})$,
%    and each $\uni^{(j)}$ is contained in $\Delta_j [-1, 1]^{m}$ with
%    $\Delta_j = 2^{-j + 1}$ --- indeed, for $t \in \uni^{(j)}$ we
%    have $\|t\|_\infty = \|\pi_{j-1}(x) - x\|_\infty \leq 2^{-j +
%    1}$.

We apply Lemma~\ref{lm:chaining-decomposition} with the norm
$\|\cdot\| = \|\cdot\|_\infty$ and $\Delta = 1$ to get the
decomposition 
$\uni \subset \uni^{(1)} + \ldots + \uni^{(k)} + \frac{\alpha}{2}
B_\infty^m$. We wish to use   Lemma~\ref{lm:decomposition} with this
decomposition, to   deduce the existence of the mechanism as in the
theorem statement.   We can use the trivial mechanism which always
outputs zero to handle the mean point problem on the universe
$\frac{\alpha}{2} B_\infty^m$ --- such a mechanism has error bounded
by $\frac{\alpha}{2}$, and satisfies $0$-zCDP. For each $\uni^{(j)}$,
with $1 \le j \leq k$, we will use the Multiplicative Weights Update
mechanism (Lemma~\ref{lm:mwu}) instantiated with privacy parameter
$\rho/k$. Hence, the privacy of the composite mechanism is $\sum_{j=1}^{k} \rho/k + 0 = \rho$.

    The total error of the mechanism guaranteed by Lemma~\ref{lm:decomposition}, is bounded by
    \begin{equation*}
        \sum_{j=1}^{k} C \left(\frac{\Delta_j (\log|\uni^{(j)}|)^{1/4} (\log m)^{1/2} k^{1/4}}{\rho^{1/4} \sqrt{n}} \right) + \frac{\alpha}{2}
    \end{equation*}
    for some universal constant $C$.

    We wish to pick $n$ large enough, so that each summand above is bounded by $\frac{\alpha}{2k}$. By rearranging, this is satisfied as soon as
    \begin{equation*}
        n \gtrsim \sup_{j \leq k} \frac{2^{-2j } \sqrt{\log \sep_\infty(\uni, 2^{-j})} k^{5/2}}{\sqrt{\rho} \alpha^2}.
    \end{equation*}
    This is equivalent, up to constant factors, to the following simpler to state condition
    \begin{equation*}
        n \gtrsim \sup_{t \geq \alpha/4} \frac{t^2 \sqrt{\log \sep_\infty(\uni, t)} (\log \alpha^{-1})^{5/2}}{\sqrt{\rho}\alpha^2},
    \end{equation*}
    which completes the proof of the theorem.
\end{proof}

\begin{proof}[Proof of Lemma~\ref{lm:projection-error}]
  Let us write $W = x + \tilde{W}$ for some mean zero,
  $\sigma$-subgaussian variable $\widetilde{W}$. Note that for any given
  realization of $\widetilde{W}$, and $W^* = \argmin_{w \in K} \|w -
  (x + \widetilde{W})\|_2$ as defined in the statement of the lemma,
  we have $\|W^* - x\|_2^2 \leq \inprod{\widetilde{W}, W^* - x}$ --- indeed, if this were not the case, then a point on a segment between $W^*$ and $x$ would be closer to $x + \widetilde{W}$, and this point belongs to $K$ by convexity.

  Therefore
  \begin{equation*}
    \E \|W^* - x\|_2^2 \leq 
    \E \inprod{\widetilde{W}, W^* - x} 
    = \E\inprod{\widetilde{W}, W^*} - \E\inprod{\widetilde{W}, x}
    = \E\inprod{\widetilde{W}, W^*}
    \leq\E \sup_{w \in K} \inprod{\widetilde{W}, t}.
  \end{equation*}
  where the final inequality follows because $W^* \in K$ by definition.
\end{proof}

\section{Packing Lower Bounds}
\label{sect:lb}

\subsection{Lower Bounds for CDP}

The starting point for the proof of Theorem~\ref{thm:lb} is the
following theorem bounding the mutual information between the input
and the output distribution of a mechanism satisfying CDP.
\begin{theorem}[\cite{BunS16}]\label{thm:mut-info-lb}
  Let $\alg$ be a mechanism satisfying $\rho$-zCDP, and let $\db$ be a
  random $n$-element database. Then,
  \[
  I(\alg(\db);\db) \le \log_2(e) \rho n^2.
  \]
\end{theorem}

Theorem~\ref{thm:mut-info-lb} and a standard packing argument give our
main lemma.

\begin{lemma}\label{lm:packing}
  For any finite $\uni \subseteq [0,1]^m$ and every $\alpha, \rho > 0$, we have
\begin{align}
  %\label{eq:lb}
  \samp^2_\rho(\uni, \alpha) &\ge \Omega \left(\frac{1}{\sqrt{\rho}} \cdot \sqrt{\log(\sep(\uni, 4\alpha))}\right);\\
  %\label{eq:lb-infy}
  \samp_\rho^\infty(\quer, \alpha) &\ge \Omega \left(\frac{1}{\sqrt{\rho}} \cdot  \sqrt{\log(\sep_\infty(\uni), 4\alpha)}\right).
\end{align}
\end{lemma}
\begin{proof}
  We give the argument for average error; the one for worst-case error
  proceeds analogously.  
  
  Take $T$ to be a $4\alpha$-separated subset of $\uni$ achieving
  $\sep(\uni, 4\alpha)$.  Take $\db$ to be a random database generated
  by picking a uniformly random element $V \in T$, and defining
  $\db$ to be the database consisting of $n$ copies of $V$. If
  $\alg$ is such that $\err^2(\quer, \alg, n) \le \alpha$, then for
  any $\db$, by Markov's inequality we have that, with probability at
  least $3/4$, $\|\alg(\db) - \bar{\db}\|_2 \le 2\alpha
  \sqrt{m}$. Let $Y$ be the closest point in $T$ to $\alg(\db)$ in
  Euclidean norm. Since any two elements $\row, \row'$ of $T$ satisfy
  $\|\row - \row'\|_2 > 4\alpha \sqrt{m}$, we have that, with
  probability at least $3/4$, $Y = V$. Then, by Fano's inequality,
  \[
  I(\alg(\db); \db) > \frac34 \log_2(|T| - 1) - 1.
  \]
  The lower bound now follows from Theorem~\ref{thm:mut-info-lb}. 
\end{proof}

Theorem~\ref{thm:lb} is a direct consequence of Lemma~\ref{lm:packing}
and the following well-known inequality, which follows by a padding
argument (see~\cite{BunUV14}):
\begin{align*}
  %\label{eq:padding}
  \forall t \ge \alpha: \ \ 
  &\samp^2_\rho(\uni, t) \le \left\lceil\nicefrac[]{\alpha}{t}\right\rceil \cdot
  \samp^2_\rho(\uni, t); \\
  &\samp^\infty_\rho(\uni, t) \le \left\lceil\nicefrac[]{\alpha}{t}\right\rceil \cdot
  \samp^\infty_\rho(\uni, t). 
\end{align*}

\subsection{Lower Bound for Local Protocols}

We prove  Theorem~\ref{thm:lb-ldp}, which shows that the sample complexity
of the Local Chaining Mechanism is optimal up to a $O(\alpha^{-2})$
factor for any $\uni$, where optimality is defined over all
sequentially interactive protocols satisfying approximate LDP.
The proof of this lower bound follows via the framework of Bassily and
Smith~\cite{BassilySmith15}, which extends the arguments
in~\cite{DJW-ASA} to approximate differential privacy.
Theorem~\ref{thm:lb-ldp} will follow from Fano's inequality and the
following lemma.

\begin{lemma}\label{lm:ldp-main-lm}
  Let $0 \le \beta \le 1$, and let $T$ be a finite subset of
  $\uni$. For any $\row \in T$ let us define the distribution
  $P^\beta_\row = \beta 1_\row + (1-\beta)U$, where $1_\row$ is the
  point distribution supported on $\row$, and $U$ is the uniform
  distribution on $T$. Let $V$ be a uniformly random sample from $T$,
  and let the elements $\row_1, \ldots, \row_n$ of the database $\db$
  be sampled independently from $P^\beta_V$. Then, for any protocol
  $\Pi$ satisfying $(\eps, \delta)$-LDP,
  \[
  I(\Pi(\db); V) = O\left(\beta^2\eps^2n + 
    \frac{\delta}{\eps}n\log |T| 
    + \frac{\delta}{\eps}n\log(\eps/\delta)\right).
  \]
\end{lemma}

Lemma~\ref{lm:ldp-main-lm} extends Theorem~1 from \cite{DJW-ASA} to
approximate LDP. It is also a (straightforward) extension of the lower
bound in \cite{BassilySmith15} to sequentially interactive mechanisms.

Before we prove Lemma~\ref{lm:ldp-main-lm}, let us show how it
implies Theorem~\ref{thm:lb-ldp}.

\begin{proof}[Proof of Theorem~\ref{thm:lb-ldp}]
  Take $T$ to be be a $t$-separated subset of of $\uni$ of size
  $\sep(\uni, t)$, for some $t \ge 6\alpha$. Let $\beta = 6\alpha /
  t$.  Observe that for any $\row, \row' \in T$, and databases $\db$
  with rows $\row'_1, \ldots, \row'_n$ sampled IID from $P^\beta_\row$,
  and $\db'$ with rows $\row_1, \ldots, \row_n$ sampled IID from
  $P^\beta_{\row'}$ we have
  \[
  \frac{1}{\sqrt{m}}\|\E\bar{\db} -  \E\bar{\db}'\|_2 \ge \beta t = 6\alpha.
  \]
  Moreover, by a standard symmetrization argument (see, e.g.~Chapter 6
  of \cite{ledoux1991probability}),
  \begin{align*}
    \E\frac1m \left\|\bar{\db} -  \E\bar{\db}\right\|_2^2
    \le 
    \frac{4}{mn^2}\E\sum_{i = 1}^n{\|\row_i\|_2^2} 
    \le \frac{4}{n},
  \end{align*}
  where in the last inequality we used the fact that $\row_i \in \uni
  \subseteq [0,1]^m$ with probability $1$. An analogous
  inequality holds for $\db'$.  By Markov's inequality, as
  long as $n \ge \frac{32}{\alpha^2}$, with probability at least
  $\frac34$ both $\bar{\db}$ and $\bar{\db'}$ are at distance at most
  $\alpha$ from their expected values, and therefore, 
  we have
  \[
  \frac{1}{\sqrt{m}}\|\bar{\db} -  \bar{\db}'\|_2 \ge  4\alpha.
  \]
  Let us assume that $\Pi$ is a protocol with error $\err^2(\uni, \Pi,
  n) \le \alpha$ for some $n \ge \frac{32}{\alpha^2}$. The inequality
  above implies that the output of $\Pi$ on a database $\db$ with $n$ rows
  sampled independently from $P^\beta_\row$ for some $\row \in T$ allows us to
  identify $\row$ with probability at least $\frac12$.
  Let $V$ be
  sampled uniformly from $T$ and let the $n$-element database $\db$ be
  sampled IID from $P^\beta_V$. By Fano's inequality, we have that 
  \[
  I(\Pi(\db); V) \ge \frac12 \log(|T| - 1) - 1 
  = \Omega(\log\sep(\uni, t)). 
  \]
  Plugging in the upper bound on $I(\Pi(\db); V)$ from
  Lemma~\ref{lm:ldp-main-lm} and solving for $n$ finishes the
  proof. 
\end{proof}

Towards proving Lemma~\ref{lm:ldp-main-lm}, we recall two lemmas from
\cite{BassilySmith15}.

\begin{lemma}[Claim 5.4 of \cite{BassilySmith15}]\label{lm:ldp-info-ineq}
  Let $Z$ be a random variable over some finite set $T$, and
  let $\alg$ be an $(\eps, \delta)$-differentially private algorithm
  defined on one-element databases from $T$. Then,
  \[
  I(\alg(Z); Z)  = O\left(\eps^2 + \frac{\delta}{\eps}\log |T| 
  + \frac{\delta}{\eps}\log(\eps/\delta)\right)
  \]
\end{lemma}
Note that above $\alg$ being differentially private simply means that
for \emph{any} two $\row, \row' \in T$, it should hold that
$\rendiv_\infty^\delta(\alg(\row)\| \alg(\row')) \le \eps$. 

We need one final lemma which allows us to get the correct dependence
on $\alpha$ in Lemma~\ref{lm:ldp-main-lm}. 

\begin{lemma}[Claim 5.5 of \cite{BassilySmith15}]\label{lm:ldp-priv-ampl}
  Let $\Pi$ be a protocol taking databases in $\uni^n$ and satisfying
  $(\eps,\delta)$-LDP. Define a protocol $\Pi'$ over databases in
  $T^n$ such that the algorithm $\Pi'_i$ run by party $i$ on private
  input $\row'_i$ and message $Y_{i-1}$ received from party $i-1$ is
  given by:
  \begin{enumerate}
  \item Sample $\row_i$ from $P^\alpha_{\row'_i}$;
  \item Run $\Pi_{i-1}(\row_i, Y_{i-1})$.
  \end{enumerate}
  Then $\Pi'$ satisfies $(O(\alpha\eps),  O(\alpha\delta))$-LDP.
\end{lemma}

\begin{proof}[Proof of Lemma~\ref{lm:ldp-main-lm}]
  Observe first that running protocol $\Pi$ on $\db$ is equivalent to
  running the protocol $\Pi'$ defined in Lemma~\ref{lm:ldp-priv-ampl}
  on the random database consisting of $n$ copies of $V$. Therefore,
  \[
  I(\Pi; V) = I(\Pi'; V) =
  I(\Pi_1', \ldots, \Pi'_n; V),
  \]
  where we use $\Pi$, $\Pi'$, $\Pi'_i$ as shorthands for $\Pi(\db)$,
  $\Pi(V, \ldots, V)$, and $\Pi_i'(V, Y_{i-1})$, respectively, and
  $Y_{i-1}$ is the message received by party $i$ from party $i-1$,
  defined to be empty for $i = 1$. By
  the chain rule of mutual information,
  \[
  I(\Pi'; V) = I(\Pi'_1; V) + I(\Pi'_2; V \mid \Pi'_1) 
  + \ldots + 
  I(\Pi'_n; V \mid \Pi'_1, \ldots, \Pi'_{n-1}). 
  \]
  By Definition~\ref{defn:ldp} and Lemma~\ref{lm:ldp-priv-ampl},
  conditional on $\Pi'_1 = (y_1, z_1), \ldots \Pi'_{i-1} = (y_{i-1},
  z_{y-1})$ for any possible views $y_1, \ldots, y_{i-1}$ and $z_1,
  \ldots, z_{i-1}$, the algorithm $\Pi'(\cdot, y_{i-1})$ satisfies
  $(O(\alpha\eps), O(\alpha\delta))$-differential privacy. Then by
  Lemma~\ref{lm:ldp-info-ineq},
  \[
  I(\Pi'_i; V \mid \Pi'_1 = (y_1, z_1), \ldots \Pi'_{i-1} = (y_{i-1}, z_{y-1})) 
  = I(\Pi'_i(V, y_{i-1}) ; V)
  = O\left(\alpha^2\eps^2 + \frac{\delta}{\eps}\log |T| + \frac{\delta}{\eps}\log(\eps/\delta)\right).
  \]
  Taking expectation over $\Pi'_1, \ldots \Pi'_{i-1}$ and summing up
  over $i$, we get 
  \[
  I(\Pi; V) =   I(\Pi'; V)
  = O\left(\alpha^2\eps^2n + \frac{\delta}{\eps}n\log |T| + 
  \frac{\delta}{\eps}n\log(\eps/\delta)\right).
\]
  This finishes the proof of the lemma.
\end{proof}

\section{Basic facts about subgaussian random variables}

The following lemma characterizes subgaussian random variables in
terms of their moments.  See Chapter~2 of
\cite{BoucheronLM-book} for a proof. 
\begin{lemma}
  \label{lem:moments-subgaussian}
  There exist universal constants $C_1, C_2$ such that
  \begin{itemize}
  \item If $W$ is a mean zero, $\sigma$-subgaussian random variable, then for every even $p$, and every $z \in \R^m$ we have $\|\inprod{W, z}\|_p \leq C_1 \sigma \sqrt{p} \|z\|$.
  \item If for every even $p$ and every $z \in \R^m$, we have $\|\inprod{W,z}\|_p \leq C_2 \sigma \sqrt{p} \|z\|$, then $W$ is $\sigma$-subgaussian.
  \end{itemize}
\end{lemma}

The following lemma bounds the parameter of the sum of independent
subgaussian random variables. Once again, the proof can be found in
Chapter 2 of \cite{BoucheronLM-book}. 
\begin{lemma}
  If $W_1, W_2, \ldots W_k$ are $\sigma_i$-subgaussian
  independent random variables, then $\sum_{i\leq k} W_i$ is $\sqrt{\sigma_1^2 + \ldots + \sigma_k^2}$-subgaussian.
\end{lemma}

The following deep fact is a consequences of Talagrand's majorizing
measures theorem~\cite{Talagrand87} (see also the
book \cite{Talagrand-book}). 
\begin{lemma}
  \label{lm:subgaussian-mw}
  For a $\sigma$-subgaussian random variable $W \in \R^m$, and an
  arbitrary compact  set $K \subset \R^m$, we have $\E \sup_{k \in K} \inprod{W, k} \lesssim \sigma g(K)$.
\end{lemma}

Finally we prove the following simple fact. 
\begin{fact}
  \label{fact:subgaussian-by-sign}
  Consider a pair of random variables $(Z, S)$ such that $Z \in \R^m$ is mean zero, $\sigma$-subgaussian random variable, and $S \in \{\pm 1\}$, not necessarily independent from $Z$. Then $SZ$ is $\Oh(\sigma)$ subgaussian.
\end{fact}
\begin{proof}
  By Lemma~\ref{lem:moments-subgaussian} it is enough to check that
  for any even $p$ and any $v \in \R^m$, we have $\|\inprod{SZ, v}\|_p \lesssim \sigma \sqrt{p}$. But we have $\|\inprod{SZ, v}\|_p = \||S|\cdot\inprod{Z, w}\|_p = \|\inprod{Z, w}\|_p\leq \sigma \sqrt{p}$.
\end{proof}

\end{document}